\documentclass[journal, 10pt]{IEEEtran}
\usepackage{amsthm}
\usepackage{amsmath}
\usepackage{mathtools} % mathtools builds on and extends amsmath package
\usepackage{amsfonts}
\usepackage{amssymb}
\usepackage{hyperref}
\usepackage{enumerate}
\usepackage{graphicx}
\usepackage{cite}
\usepackage{url}

\usepackage[dvipsnames]{xcolor}
\usepackage{bm} % bold symbols in math mode

	\hypersetup{%		
   		pdfborder = {0 0 0},  % disable borders around links
    	colorlinks,			  % sets the colors for example in the TOC
    	citecolor=black,
    	filecolor=black,
    	linkcolor=black,
    	urlcolor=cyan!50!black!90
	}
\usepackage[xindy, toc, acronyms]{glossaries} %  , nonumberlist   
\newglossary[slg]{symbolslist}{syi}{syg}{Symbols} % create add. symbolslist
\makeglossaries
\newacronym{GMTI}{GMTI}{ground moving target indication}
\newacronym{STAP}{STAP}{space-time adaptive processing}
\newacronym{SAR}{SAR}{synthetic aperture radar}
\newacronym{PIRDIS}{PIRDIS}{Platform Independent Range Doppler Image Simulation}

\newacronym{MIMO}{MIMO}{multiple-input multiple-output}
\newacronym{SIMO}{SIMO}{single-input single-output}
\newacronym{AESA}{AESA}{active electronically scanned array}
\newacronym{TDM}{TDM}{time-division multiplexing}

\newacronym[longplural={cells under test}]{CUT}{CUT}{cell under test}
\newacronym{CS}{CS}{compressed sensing}
\newacronym{FFT}{FFT}{fast Fourier transform}
\newacronym{DTFT}{DTFT}{discrete-time Fourier transform}

\newacronym{SNR}{SNR}{signal-to-noise ratio}
\newacronym{RCS}{RCS}{radar cross section}

\newacronym{MSE}{MSE}{mean squared error}
\newacronym{CRB}{CRB}{Cram\'er-Rao bound}

\newacronym{BMSE}{BMSE}{Bayesian mean squared error}
\newacronym{BCRB}{BCRB}{Bayesian Cramer-Rao bound}
\newacronym{ZZB}{ZZB}{Ziv-Zakai bound}
\newacronym{WWB}{WWB}{Weiss-Weinstein bound}
\newacronym{BZB}{BZB}{ Bobrovsky-Zaka\"i bound}

\newacronym{RP}{RP}{random-phase}
\newacronym{KP}{KP}{known-phase}
\newacronym{UC}{UC}{unconditional}

\newacronym{PRI}{PRI}{pulse repetition interval}
\newacronym{PRF}{PRF}{pulse repetition frequency}
\newacronym{CPI}{CPI}{coherent processing interval}
\newacronym{DoA}{DoA}{direction of arrival}
\newacronym{FoV}{FoV}{field of view}

\newcommand{\expec}[1]{\mathbb{E}_{#1}}

\newcommand{\Gaussian}[1]{\mathcal{N}_{#1}}
\newcommand{\Uniform}[1]{\mathcal{U}({#1})}

\makeatletter  % generic probablity density which can be conditional or not
\def\pd{\@ifnextchar[{\@pdwith}{\@pdwithout}}
\def\@pdwith[#1]#2{p(#1 \,\vert\, #2)}
\def\@pdwithout#1{p(#1)}
\makeatother

\newcommand{\obs}{\bm{x}}
\newcommand{\Obsdom}{\Omega}

\newcommand{\numobs}{N}

\makeatletter  
\def\pargen{\@ifnextchar[{\@pargenwith}{\@pargenwithout}} % generic parameter
\def\@pargenwith[#1]{\theta_{#1}}
\def\@pargenwithout{\bm{\theta}}
\def\Pargen{\@ifnextchar[{\@Pargenwith}{\@Pargenwithout}} % generic parameter domain
\def\@Pargenwith[#1]{\Theta_{#1}}
\def\@Pargenwithout{\Theta}
\makeatother

\newcommand{\pargenh}{\hat{\bm{\theta}}}
\newcommand{\numpar}{q}

\newcommand{\sensingpol}{\bm{g}}
\newcommand{\gall}[1]{\sensingpol^{(#1)}}
\newcommand{\gone}[1]{\sensingpol_{#1}}

\newcommand{\sv}{\bm{a}}

\newcommand{\gmodel}{\bm{a}}
\newcommand{\map}[3]{#1:#2 \rightarrow #3}

\newcommand{\ramp}{r}
\newcommand{\costk}{\cost_k}

\newcommand{\until}[1]{\{1,\dots,#1\}}

\newcommand{\pargenkm}{\pargen_{k-1}}
\newcommand{\pargenk}{\pargen_k}

\newcommand{\obskmall}{\obs^{(k-1)}}
\newcommand{\obskall}{\obs^{(k)}}

\newcommand{\obsk}{\obs_{k}}

\newcommand{\pplusatkm}{p^+_{k-1}}

\newcommand{\pplusatk}{p^+_{k}}
\newcommand{\pminusatk}{p^-_{k\textsl{}}}

\newcommand{\pminus}{p^{-}}
\newcommand{\pplus}{p^{+}}

\newcommand{\pzeroplus}{p^{+}_0}

\newcommand{\sensingpolk}{\gone{k}}

\newcommand{\sensingpolkall}{\gall{k}}

\newcommand{\sensingpolkmall}{\gall{k-1}}

\newcommand{\noise}{\bm{n}}
\newcommand{\noisecov}{\bm{R}}

\makeatletter  % parameter upper and lower bounds
\def\pargenL{\@ifnextchar[{\@pargenLwith}{\@pargenLwithout}}
\def\@pargenLwith[#1]{a_{#1}}
\def\@pargenLwithout{\bm{\alpha}}
\def\pargenU{\@ifnextchar[{\@pargenUwith}{\@pargenUwithout}}
\def\@pargenUwith[#1]{b_{#1}}
\def\@pargenUwithout{\bm{\beta}}
\makeatother

\newcommand{\transitionden}{p}

\newcommand{\NP}{N_P}

\makeatletter  
\def\partPF{\@ifnextchar[{\@partPFwith}{\@partPFwithout}}
\def\@partPFwith[#1]{\bm{p}_{#1}}
\def\@partPFwithout{\bm{p}}

\def\weightPF{\@ifnextchar[{\@weightPFwith}{\@weightPFwithout}}
\def\@weightPFwith[#1]{{w_{#1}}}
\def\@weightPFwithout{\bm{w}}

\def\meanPF{\@ifnextchar[{\@meanPFwith}{\@meanPFwithout}}
\def\@meanPFwith[#1]{{\bar{p}_{#1}}}
\def\@meanPFwithout{\bm{\bar{p}}}

\def\varPF{\@ifnextchar[{\@varPFwith}{\@varPFwithout}}
\def\@varPFwith[#1]{{\hat{\sigma}^2_{#1}}}
\def\@varPFwithout{\bm{\hat{\sigma}^2}}

\makeatother

\newcommand{\doa}{u}
\newcommand{\amplitude}{s}
\newcommand{\dop}{\omega}
\newcommand{\velr}{v_r}
\newcommand{\numRx}{N_{\rm{Rx}}}
\newcommand{\numTx}{N_{\rm{Tx}}}
\newcommand{\svtdm}{\bm{b}}

\newcommand{\rsm}{\bm{G}^{\rm{Rx}}}
\newcommand{\tsm}{\bm{G}^{\rm{Tx}}}
\newcommand{\dRx}{\bm{d}^{\rm{Rx}}}
\newcommand{\dTx}{\bm{d}^{\rm{Tx}}}

\newcommand{\dVirt}{\bm{d}^{\rm{V}}}

\newcommand{\numRxsel}{N_{\rm{R}}}

\newcommand{\numPulse}{N_{\rm{P}}}
\newcommand{\tPulse}{\bm{t}^P}
\newcommand{\tVirt}{\bm{t}^{\rm{V}}}

\newcommand{\freq}{u}
\newcommand{\freqU}{b}
\newcommand{\freqL}{a}
\newcommand{\arr}[1]{\bm{d}_{#1}}
\newcommand{\arrex}{\bm{D}}
\newcommand{\magnitude}{\vert s\vert}
\newcommand{\phase}{\varphi}
\newcommand{\SNR}{\gamma}
\newcommand{\SNRUC}{\gamma'}
\newcommand{\noisevar}{\sigma^2}
\newcommand{\scaling}{g}
\newcommand{\BPh}{h}

\newcommand{\Freq}{\bm{\freq}}
\newcommand{\Freqmean}{\bm{\mu_{\freq}}}
\newcommand{\Freqvar}{\bm{\Sigma_{\freq}}}
\newcommand{\Phase}{\Phi}
\newcommand{\Phaseis}{\tilde{\Phi}}
\newcommand{\BC}{\operatorname{BC}}
\newcommand{\intprior}{\xi}
\newcommand{\tpDom}{\mathcal{H}_{\Pargen}}

\newcommand{\freqhat}{\hat{\freq}}
\newcommand{\emse}{\mathbf{\rm{MSE}}}
\newcommand{\Ntrial}{N_{T}}

\makeatletter  
\def\arrI{\@ifnextchar[{\@arrIwith}{\@arrIwithout}} % intersection of parameter domains
\def\@arrIwith[#1]{d_{#1}}
\def\@arrIwithout{\bm{d}}
\makeatother

\newcommand{\arrICM}{\hat{d}}

\newcommand{\bmse}{{\rm{\mathbf{BMSE}}}}

\newcommand{\FIM}{{\rm{\mathbf{FIM}}}}

\newcommand{\cost}{{\mathcal{C}}}

\makeatletter  
\def\weighting{\@ifnextchar[{\@weightingwith}{\@weightingwithout}} % trace weighting
\def\@weightingwith[#1]{\rho_{#1}}
\def\@weightingwithout{\bm{\rho}}
\makeatother
\newcommand{\wwb}{ {\rm{\mathbf{WWB}}}}

\newcommand{\tpH}{\bm{H}}
\newcommand{\err}{\bm{e}}

\newcommand{\tph}[1]{\bm{h}_{#1}}

\newcommand{\Q}{\bm{Q}}
\newcommand{\QI}{Q}
\newcommand{\tpgenI}{\bm{v}}
\newcommand{\tpgenII}{\tilde{\bm{v}}}
\newcommand{\USNR}{\kappa}

\makeatletter  
\def\Pargenis{\@ifnextchar[{\@Pargeniswith}{\@Pargeniswithout}} % intersection of parameter domains
\def\@Pargeniswith[#1]{\tilde{\Theta}_{#1}}
\def\@Pargeniswithout{\tilde{\Theta}}

\def\tphI{\@ifnextchar[{\@tphIwith}{\@tphIwithout}} % single column test point
\def\@tphIwith[#1]{h_{#1}}
\def\@tphIwithout{\bm{h}}
\makeatother

\newcommand{\repart}[1]{\operatorname{Re}\{#1\}}
\newcommand{\diag}{\operatorname{diag}}
\newcommand{\define}{:=}
\newcommand{\enifed}{=:}

\newcommand{\leqm}{\preceq}

\newcommand{\traceW}{\operatorname{trace}_{\weighting}}
\newcommand{\setdef}[2]{\{#1 \, : \, #2\}}
\newcommand{\ones}[1]{\bm{1}_{#1}}
\newcommand{\eye}[1]{\bm{I}_{#1}}
\newcommand{\vol}[1]{\vert #1\vert}
\newcommand{\cf}[1]{{\chi}_{#1}}

\newcommand{\R}{{\mathbb{R}}}
\newcommand{\Rpos}{\mathbb{R}_{\ge 0}}
\newcommand{\C}{{\mathbb{C}}}
{}
{}
\newtheorem{proposition}{Proposition}{}
\newtheorem{remark}{Remark}{}
\newtheorem{lemma}{Lemma}{}
\newtheorem{algorithm}{Algorithm}{}

\newcommand{\margin}[1]{\marginpar{\color{blue}\tiny\ttfamily#1}}

\newcommand{\cg}[1]{}

\begin{document}

\IEEEspecialpapernotice{\small{
%This paper is a preprint of a paper submitted to IET Radar, Sonar \& Navigation. %, Special Issue on Cognitive Radar.
%	If accepted, the copy of record will be available at the IET Digital Library.	
This paper is a preprint of a paper accepted by IET Radar, Sonar \& Navigation %, Special Issue on Cognitive Radar,
 and is subject to Institution of Engineering and Technology Copyright. 
 When the final version is published, the copy of record will be available at the IET Digital Library	
	}}

\title{Adaptive transmission for radar arrays\\ using Weiss-Weinstein bounds}

\author{Christian Greiff, David Mateos-N\'u\~nez,
		 Mar\'ia A. Gonz\'alez-Huici, Stefan Br\"uggenwirth
\thanks{The authors are with the Department of Cognitive Radar, Fraunhofer FHR, Wachtberg, Germany. Email: christian.greiff@fhr.fraunhofer.de}
}


\maketitle

\begin{abstract}
We present an algorithm for adaptive selection of pulse repetition frequency or antenna activations for Doppler and DoA estimation. The adaptation is performed sequentially using a Bayesian filter, responsible for updating the belief on parameters, and a controller, responsible for selecting transmission variables for the next measurement by optimizing a prediction of the estimation error. This selection optimizes the Weiss-Weinstein bound for a multi-dimensional frequency estimation model based on array measurements of a narrow-band far-field source. A particle filter implements the update of the posterior distribution after each new measurement is taken, and this posterior is further approximated by a Gaussian or a uniform distribution for which computationally fast expressions of the Weiss-Weinstein bound are analytically derived.
We characterize the controller's optimal choices in terms of SNR and variance of the current belief, discussing their properties in terms of the ambiguity function and comparing them with optimal choices of other  Weiss-Weinstein bound constructions in the literature. The resulting algorithms are analyzed in simulations where we showcase a practically feasible real-time evaluation based on look-up tables or small neural networks trained off-line.
\end{abstract}

\section{Introduction}\label{sec:Introduction}

Software-defined radar systems offer degrees of freedom such as waveform diversity, beam-steering, or antenna selection. 
The concept of Cognitive Radar \cite{SH:12} describes a dynamic systems approach for the control of these degrees of freedom in a real-time, closed-loop fashion, %This novel approach has sparked great interest in both civil and military domains.
motivating research on waveform design~\cite{gini2012waveform, wicks2011principles, blunt2016overview}, matched illumination~\cite{romero2011theory,goodman2007adaptive},
radar resource management~\cite{charlish2015phased,ding2008survey}, or spectral coexistence~\cite{aubry2014cognitive, stinco2016spectrum}. Application areas include multi-functional \gls{AESA} systems in airborne or maritime scenarios~\cite{guerci2014cofar, smits2008cognitive}, automotive \gls{MIMO}-radars~\cite{WH-JT-RS:13,OI-JT-IB:15, JT-OI-IB:16} %, patentSDR-CD-MM:17}
 or distributed sensor networks~\cite{chavali2012scheduling}.

The \textit{perception-action} cycle~\cite{SH:06, KLB-CJB-GES-JTJ-MR:15b} describes a sequential process of extracting information from a scene and using that knowledge for adapting the transmission and processing of subsequent measurements. These are tasks of estimation and control that can be formulated in a Bayesian estimation framework~\cite{HLVT-KLB:13},~\cite{SH:12}. Sequential estimation is performed by updating a belief distribution over the parameter of interest according to motion and measurement models, which is suitable for tracking and Track-Before-Detect approaches. Practical implementations such as Particle Filters \cite{MSA-SM-NG-TC:02} or Cubature Kalman Filters~\cite{SH:12} can handle nonlinear models crucial for radar systems. On the control side, the estimation performance can be predicted and optimized
using Bayesian bounds~\cite{HLVT-KLB:13} that provide a lower bound on the expectation over the current belief of the covariance matrix of the error of any estimator. Classical examples include the \gls{BCRB}, and other members of the family of \glslink{WWB}{Weiss-Weinstein bounds (WWB)}~\cite{AW-EW:85,AR-PF-PL-CDR-AN:08}, including the \gls{BZB}. The latter, together with the \gls{ZZB}~\cite{DK-KLB:10}, take into account estimation errors in nonlinear estimation problems that occur below a certain \gls{SNR} due to sidelobes in the ambiguity function~\cite{NDT-AR-RB-SM-PL:12}. These are underestimated by the Cram\'er-Rao bound (CRB), which is related to the mainlobe-width and therefore measures the accuracy and resolution limit when the \gls{SNR} is large. Instances of such nonlinear estimation problems include frequency estimation in radar systems such as the estimation of Doppler frequency and \gls{DoA}.  These scenarios have been the context for adaptive strategies for the selection of pulse repetition frequency~\cite{KLB-JTJ-GES-CJB-MR:15a} using the \gls{BCRB}, and for antenna selection using the \gls{BZB}~\cite{OI-JT-IB:15,JT-OI-IB:16} and \gls{WWB}~\cite{DMN-MGH-RS-SB:17}.

The optimization metrics prescribed by the WWB are themselves an optimization over so-called test points of an expression that contains integrals of the measurement and prior distributions. These probability models affect the characterization of optimal sensing parameters, and also the existence of analytical expressions for the integrals. 
In the context of sequential estimation of a dynamic Markovian parameter, these metrics can be constructed to bound the \gls{BMSE} given all previous measurements. 
In this case, the motion and measurement updates can be embedded in a sequential computation of the \gls{WWB} under an ample class of dynamics~\cite{IR-YO:07}, resulting in explicit formulas for some families of prior and measurement distributions~\cite{FX-PG-GM-CFM:13}.
 Alternatively, the works~\cite{WH-JT-RS:13,OI-JT-IB:15,NS-JT-HM:15} show that under exact posteriors, the Bayesian bound becomes too a conditional bound given previous measurements, and suggest approximating these posteriors with a particle filter or the Metropolis-Hastings algorithm.
In this work we use a particle filter to update the posterior, that is further approximated by a combination of Gaussian and uniform distributions for the selection of sensing parameters using the \gls{WWB}. For this purpose and those priors, we
 have derived the \gls{WWB} 
% with Gaussian and independent uniform priors 
 for frequency estimation using array measurement models for a single source with known \gls{SNR} but random initial phase, with a rigorous derivation of test-point domains.  This model can be particularized to azimuth and elevation estimation, azimuth and Doppler estimation in \gls{TDM} \gls{MIMO}.
Modeling-wise, in the case of \gls{DoA} or Doppler estimation, we characterize the \gls{WWB}-based optimal selection of a scaling parameter that models the carrier frequency or the \gls{PRI}, in terms of the \gls{FoV} or variance of the prior density, 
%which is a key input for adaptive Bayesian algorithms,
 and compare the benefits of the proposed model with alternative \gls{WWB} constructions called \textit{conditional}, which we refer to as \textit{known-phase}, and \textit{unconditional} signal models~\cite{DTV-AR-RB-SM:14}, demonstrating the influence of modeling the initial phase as random while regarding the \gls{SNR} as deterministic and known.
Computationally, this formulation has the advantage of fast, vectorized evaluation over test-points thanks to explicit formulas without needing to evaluate the inverse of a matrix. 
We then show in simulations the closed-loop performance of the particle filter combined with the above criteria using feedback on the variance of the posterior for adaptation of \gls{PRF} or array scaling, and for antenna selection.

%\vspace*{-0.75cm}
\paragraph*{Organization}

The paper is organized as follows:
Section \ref{sec:adaptive-sensing-framework} describes the adaptive sensing algorithm for sequential Bayesian estimation.
Section~\ref{sec:wwb-and-models-random-phase-array-processing} presents a derivation of the \gls{WWB} for general array processing tasks for a single source of known \gls{SNR} under a spatio-temporal sampling scheme with a random initial phase associated to the coherent processing interval.
Section~\ref{sec:controller-wwb} includes an analysis of the controller choices for several \gls{WWB} models and analyzes the consequences of assuming knowledge of the initial phase or lack thereof. 
Section~\ref{sec:adaptive-array-scaling} applies the general framework to the problem of adaptive \gls{PRF} or array scaling for Doppler or \gls{DoA} estimation, and to the problem of channel selection for \gls{DoA} estimation. The resulting adaptive policies are compared in simulations, exemplifying the practical implementation of our strategies with the use of look-up tables and neural nets.
Section~\ref{sec:Conclusion} discusses our conclusions and ideas for future work, and we include Appendices with auxiliary results.

\paragraph*{Notational conventions}
$\R^{\numobs}$ and $\C^\numobs$ denote the $\numobs$-dimensional real and complex vector spaces, respectively.
The real part of a complex number $z\in\C$ is $\repart{z}$, while $|z|$ stands for the absolute value. Likewise, the Euclidean volume for sets $\Pargen\subset{\R^\numpar}$ is denoted by $\vol{\Pargen} = \int \cf{\Pargen}(\pargen) d\pargen$, where $\cf{\Pargen}$ is the indicator function. 
For a symmetric real or Hermitian complex matrix $\bm{A}$, the induced norm is $\Vert \obs \Vert_{\bm{A}}\define \sqrt{\obs^H \bm{A} \obs }$ where $\obs^H$ is the conjugated transpose of vector $\obs$.
The weighted trace is defined as $\traceW (\bm{A}) \define \sum_i \weighting[i] A_{ii}$.
%for a diagonal matrix $\bm{W} = \diag((w_i)_i)$. 
%
We denote by $\ones{N}\in\R^N$ the vector of ones and $\eye{N}\in\R^{\numobs\times\numobs}$ the identity matrix.
For functions, $\map{p,f}{\R^q}{\R}$, we define the expectation of $f$ with respect to the density $p$ as  $\expec{p(\pargen)}[f]=\int f(\pargen) p(\pargen) d\pargen$.
Bracketed integer superscripts serve as abbreviation for a collection of variables, e.g. $\obskall = \{\obs^1, \dots, \obs^k\}$.
\\
We adhere to the convention that lower case letters denote scalars (e.g. $\SNR\in\R$), lower case boldface letter denote vectors (e.g. $\arrI\in \R^\numobs$) and upper case boldface letters denote matrices (e.g. $\tpH\in \R^{\numpar\times M}$). 
We frequently use a Matlab inspired notation for vector evaluations of a scalar function, e.g. for a vector $\arrI = (\arrI[n])_{n = 1}^\numobs\in\R^\numobs$, the expression $e^{i\arrI} = (e^{i\arrI[n]})_{n = 1}^\numobs \in\C^{\numobs}$ also denotes a vector.

%---------------------------------------- Adaptation Framework
\section{Bayesian adaptive sensing for sequential estimation}\label{sec:adaptive-sensing-framework}

Here we describe a framework for sequential adaptive sensing using as feedback a belief distribution of the parameter of interest. The control system comprises a \textit{processor}, which uses a Bayesian filter to incorporate information from measurements about a parameter of interest into the belief distribution, and a \textit{controller}, which is a rule for selecting the transmission variables for the next measurement using feedback on the current knowledge. Next we describe the processor and the controller.

%\paragraph{Processor}
\subsection{Processor}
The processor is in charge of incorporating information from the latest measurement into the belief distribution of the parameter of interest.
Consider a parameter vector that at time $k$ is modeled by the random vector $\pargenk\in\R^\numobs$. 
To relate the measurement $\obsk$ at step $k$ with the parameter $\pargenk$,  we need a measurement model, $\pd[\obsk]{\pargenk, \sensingpolk}$, that depends on the sensing parameters $\sensingpolk$ used in that measurement. In Section~\ref{sec:wwb-and-models-random-phase-array-processing}, we substantiate this model focusing on multidimensional frequency estimation of a single complex sinusoid with additive Gaussian white noise where the sensing parameters are sampling schemes in time and space. We keep this section general for suitable distributions $\pd[\obsk]{\pargenk, \sensingpolk}$.

%\paragraph{Scene as perceived by sensor}
To model the evolution between measurement steps of the parameter being estimated, we assume a Markovian \textit{transition} (or state evolution) model of the form $ \pd[\pargenk]{\pargen_{k-1}, \sensingpolk}$. Note that in general there can be a dependence on the sensing parameter, e.g., if  the latter specifies the time of the measurement at step $k$. This transition probability is assumed known.

A Bayesian filter proceeds in two steps: the motion update (or prediction) and the measurement update (or filtering). The initial belief distribution for the parameter, denoted by $p(\pargen_0)=p_0(\pargen_0)$ is a modeling choice.

\paragraph*{Motion update} The motion update predicts the \textit{state} of the parameter at the time of the next measurement using the model for state evolution.
Suppose that after measurement step $k-1$, we have a belief given by $\pplusatkm(\pargen_{k-1})$. Then the motion update of the belief distribution is given by the Chapman-Kolmogorov equation \cite{MSA-SM-NG-TC:02},
\begin{align}\label{eq:motion-update-general}
%	\pminusatk(\pargenk\vert \sensingpolk) \define \int \pd[\pargenk]{\pargenkm, \sensingpolk}  \;\pplusatkm(\pargenkm) d\pargenkm.
		\pminusatk(\pargenk) \define \int \pd[\pargenk]{\pargenkm, \sensingpolk}  \;\pplusatkm(\pargenkm) d\pargenkm.
\end{align}
This probability is employed as predicted \textit{prior} belief at step~$k$.

\paragraph*{Measurement update} The measurement update filters the prediction using the likelihood of the measurement,
\begin{align}\label{eq:measurement-update-general}
%\pplusatk(\pargenk\vert \sensingpolk)  &\define c\;  \pd[\obsk]{\pargenk, \sensingpolk} \;\pminusatk(\pargenk\vert \sensingpolk)  .
\pplusatk(\pargenk)  &\define c\;  \pd[\obsk]{\pargenk, \sensingpolk} \;\pminusatk(\pargenk) 
\end{align}
with $c$ chosen so that $\pplusatk$ is a probability density over values of $\pargenk$, and $\pzeroplus(\pargen_0)\define p_0(\pargen_0)$.
The recurrences \eqref{eq:motion-update-general} and \eqref{eq:measurement-update-general} have properties that can depend on the policy for sensing parameters at the Controller.

\begin{figure}[]
	%\hspace*{-0.75cm}
	\hspace*{-0.4cm}
	%	\subfigure[Title 2]
	{\includegraphics[width=1.05\linewidth]{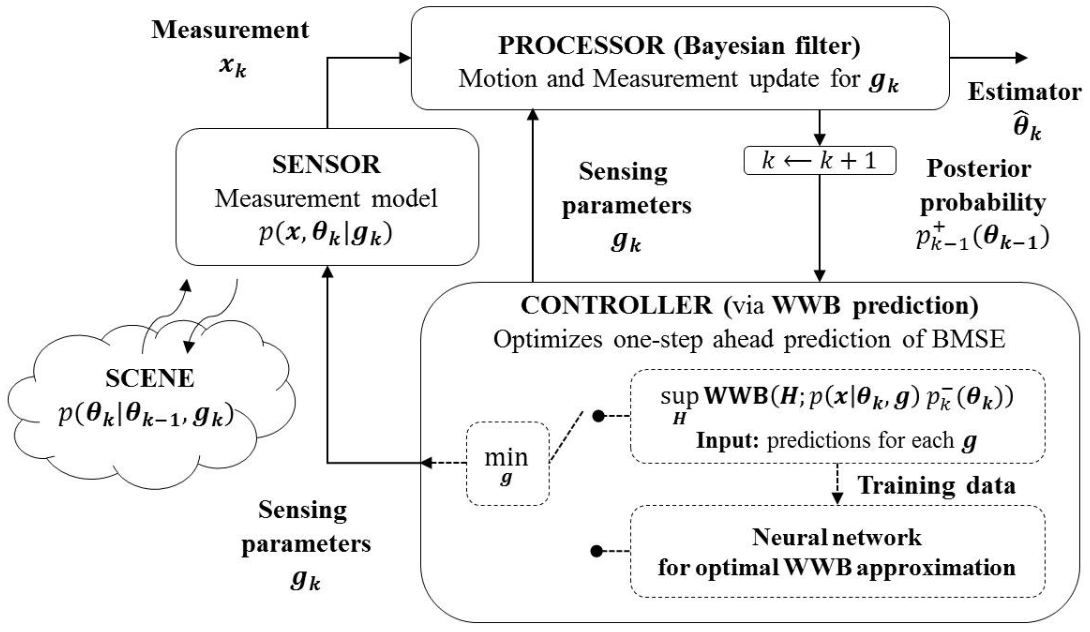}}
%		{\includegraphics[width=1.05\linewidth]{figures/generalfigures/diagram_nn_sensing.jpg}}
	%	\subfigure[Title 2]
	%	{\includegraphics[width=.48\linewidth]{}
	%	\begin{picture}(0,0)(0,0)
	%	%\put(-93,42){$$}  
	%	\end{picture}
	\caption{Diagram of adaptive sensing based on the \gls{WWB}. This Bayesian framework is analogous to the ones in~\cite{KLB-CJB-GES-JTJ-MR:15b, DMN-MGH-RS-SB:17}, where we consider in this work the optional choice of training a neural network for \gls{WWB} ranking of candidate sensing parameters. 
}\label{fig:adaptive-policy-diagram}
\end{figure}

%\paragraph{Controller} 
\subsection{Controller} 
The controller is in charge of selecting sensing parameters for the next measurement using as input the current belief distribution. That is, at step $k$ it takes as input the posterior of the last step  $\pplusatkm(\pargenkm)$ before the motion update, or approximation thereof, and returns $\sensingpolk$. Any criterion to make this selection should employ the state evolution model $\pd[\pargenk]{\pargen_{k-1}, \sensingpolk}$ and the observation model %\eqref{eq:general-gaussian-model}
according to candidate sensing parameters $\pd[\obsk]{\pargenk, \sensingpolk}$.
The criterion used in this paper is a tight lower bound on the \gls{BMSE}.

For the data model $p(\obs,\pargen)$, the \glslink{BMSE}{\textit{Bayesian Mean Squared Error} (BMSE)} of an estimator $\pargenh\equiv\pargenh(\obs)$ of $\pargen$ is defined as the Bayesian covariance matrix of the error $\err \define \pargenh(\obs) - \pargen$,~i.e.,
\begin{align}\label{eq:bmse-general}
\bmse(\pargenh; \pd{\obs,\pargen})\define \expec{\pd{\obs,\pargen}} \big[(\pargenh(\obs)-\pargen)(\pargenh(\obs)-\pargen)^T\big].
\end{align}
The \gls{BMSE} in~\eqref{eq:bmse-general} can be used as optimization metric for adaptive sensing~\cite{SH-JRH-BM:18}, but is expensive to compute because it involves Monte Carlo integrals over the parameter space and over realizations of the observation. 
Instead, we follow the common practice of replacing the BMSE by one of its lower bounds, e.g. the \glslink{WWB}{\textit{Weiss-Weinstein bound} (WWB)}. 
The \gls{WWB} provides a lower bound on the \gls{BMSE} of any estimator and thus gives an indication of the achievable estimation performance. 
Formally, the bound is obtained from a covariance inequality in the sense of the Loewner order on positive semi-definite matrices~\cite[p. 333]{HLVT-KLB:13} as
\begin{align}\label{eq:ineq-bmse-wwb-general}
\wwb(\tpH;\pd{\obs,\pargen}) \leqm\bmse(\pargenh;\pd{\obs,\pargen}) ,
\end{align}
where $\wwb(\tpH;\pd{\obs,\pargen})\in \R^{\numpar\times \numpar}$ is a member of the family of \glspl{WWB} parametrized by the \textit{test point matrix} $\tpH$ for a data model $\pd{\obs,\pargen}$ as described in the Appendix \ref{sec:WWB-theory}.

Depending on the estimation task, we are interested in the contribution to the \gls{BMSE} of a subset 
of coordinates,
%of the estimation errors, 
and thus we define the following objective function for candidate sensing parameters,
	\begin{align}\label{def:costwwb}
\costk(\sensingpol) \define \sup_{\tpH} \traceW (\wwb(\tpH; \pd[\obsk]{\pargenk,\sensingpol} \pminusatk(\pargenk))) ,
	\end{align}
where the predicted prior $ \pminusatk(\pargenk) \equiv \pminusatk(\pargenk; \sensingpol)$ depends on $\sensingpol$ and is obtained from the input $\pplusatkm(\pargenkm)$ through the motion update~\eqref{eq:motion-update-general}.
The weighting vector $\weighting\in\Rpos^\numpar$ can be used to balance units or weight
 the components of $\pargen$.
 Optimization over test points $\tpH$ is performed to obtain the tightest bound within the parametric family of bounds.
 The sensing parameters are then found as 
%the solution of the following optimization,
\begin{align}\label{eq:opt-sensing}
\sensingpolk = \arg \min_{\sensingpol} \costk(\sensingpol). 
\end{align}
This selection requires a double optimization procedure, first over test point matrix $\tpH$, to evaluate the prediction of the \gls{BMSE}, and then over sensing parameters. The former is non-convex and we use a global optimization algorithm (e.g. simulated annealing~\cite{SK-CDG-MPV:83}). A visualization of the closed-loop between the processor and the controller is depicted in Fig.~\ref{fig:adaptive-policy-diagram}. Algorithm \ref{alg:adaptivesensing} summarizes the steps. 
%for sensing parameter 
%adaptation under the assumption of Remark \ref{rem:motionG}. 
\\

%\fbox{\parbox[c]{8.0cm}{

\begin{algorithm}[Adaptive selection of sensing parameters]\label{alg:adaptivesensing}
	\mbox{}\\
	
	\textbf{Input:} Initial belief distribution $\pplus_0(\pargen)$;
	Measurement likelihood model $\pd[\obs]{\pargen, \sensingpol}$;
	State evolution model $p(\pargenk |\pargen_{k-1})$
	\\

	\textbf{Output:} Belief distribution $\pplusatk(\pargenk \vert \sensingpolk)$\\ %$\pplusatk(\pargenk \vert \sensingpolk)$
%	and estimation $\pargenh_k$ of current parameter $\pargenk$\\
	
	\textbf{Procedure:}
	Set $k = 1$. 
	\begin{enumerate}[\hspace{0.3cm} 1. ]
		\item \label{alg:startloop} Motion update of belief distribution via~\eqref{eq:motion-update-general} 
		%(c.f. Remark \ref{rem:motionG})
		to obtain $\pminusatk(\pargenk)$
		\vspace{3mm}
		\item The controller finds "optimal" sensing parameters $\sensingpolk$ by minimizing the cost function
		$\sensingpolk = \arg \min_{\sensingpol} \costk(\sensingpol)$
		\vspace{3mm}
		\item Measurement is performed, yielding observation $\obsk$
		\vspace{3mm}
		\item The processor performs the measurement update of belief distribution via~\eqref{eq:measurement-update-general} to obtain the posterior $\pplusatk(\pargenk)$
		\vspace{3mm}
		\item Start next cycle with $\pplusatk(\pargenk)$ % $\pplusatk(\pargenk \vert \sensingpolk)$ 
		as new initial belief by increasing $k \leftarrow k+1$ and repeating from step 1
	\end{enumerate}
\end{algorithm}
%}}
%\\

\begin{remark}[Dependence of motion model on sensing parameter ]\label{rem:motionG}
	
	Note that if the motion update depends on the sensing parameter (e.g., if it refers to the time of measurement), then the controller has to perform the motion update in~\eqref{eq:motion-update-general} for each evaluation of the cost function $\costk(\sensingpol)$. For the special case of a $\sensingpol$-independent state evolution model, the motion update in~\eqref{eq:motion-update-general} needs to be performed only once before passing the resulting prediction %$\pminusatk$ 
	to the controller.
\end{remark}

\subsection{Legitimation of the closed-loop}

The closed-loop formed by the recurrences \eqref{eq:motion-update-general} and \eqref{eq:measurement-update-general} and the selection of sensing parameters~\eqref{eq:opt-sensing} governs the evolution of the belief distribution of the parameter.
Under this evolution, the cost function is a lower bound of the BMSE conditional to previous measurements $\obskall\define \{\obs_1, \dots, \obsk \}$~\cite{WH-JT-RS:13,OI-JT-IB:15,NS-JT-HM:15}.

\begin{proposition}[Properties of the closed-loop]\label{cor:CBMSE}
	The following relations hold for the closed-loop system formed by the Processor updates \eqref{eq:motion-update-general} and \eqref{eq:measurement-update-general}, and the Controller selection~\eqref{eq:opt-sensing}.\vspace*{-0.25cm}
\begin{itemize}	
\item[(i)] The motion and measurement updates satisfy
		\begin{subequations}\label{eq:updates-identities}
		\begin{align}
		%\pminusatk(\pargenk\vert \sensingpolk) 
		\pminusatk(\pargenk) 
		&= \pd[\pargenk]{\obskmall, \sensingpolkall} 	
		\label{eq:updates-identities-1}
		\\
		%\pplusatk(\pargenk\vert \sensingpolk)
		\pplusatk(\pargenk)   
		& =  
		\pd[\pargenk]{\obskall, \sensingpolkall} 
		\label{eq:updates-identities-2}
		\\
		c 
		& = \frac{1}{ \pd[\obsk]{\obskmall,\sensingpolkall} } ,
		\label{eq:updates-identities-3}
		\end{align}
	\end{subequations}
i.e., $\pplusatk(\pargenk)$ is the \textit{posterior} belief at step $k$ conditioned 
to all measurements $\obskall$ and sensing parameters $\sensingpolkall$.

\item[(ii)] The \gls{WWB} in the cost function~\eqref{def:costwwb} satisfies the inequality
	\begin{align*}%\label{eq:bound-all-previous-measurements}
	%&\wwb(\tpH; \pd[\obsk]{\pargenk,\sensingpol}\pminusatk(\pargenk|\sensingpol)) \\
	&\wwb(\tpH; \pd[\obsk]{\pargenk,\sensingpol}\pminusatk(\pargenk)) \\
	&
	\qquad\leqm 
	%		\bmse(\pargenh,\pd[\obs]{\pargen,\sensingpol}, \pd[\pargenk]{\obskmall, \sensingpolkall})
	%		\\
	%		&=
	\bmse(\pargenh_k,\pd[\obsk,\pargenk]{\obskmall, \sensingpolkmall, \sensingpol}) ,
	\end{align*}
	%
	%		\begin{align}\label{eq:ineq-bmse-wwb-general}
	%		\bmse(\pargenh,\pd{\obs,\pargen}) \geqm \wwb(\tpH,\pd{\obs,\pargen}) ,
	%		\end{align}
	where $\pargenh_k $ is any estimator based 
	on $\obskall, \sensingpolkmall, \sensingpol$.
%	on the Bayesian update $\pd[\pargenk]{\obskall, \sensingpolkmall, \sensingpol}$. 
	Similarly for the corresponding inequality taking $\traceW$ on both sides.
%	and in particular,
%	\begin{align*}
%	\costk(\sensingpol) \le \sum_{i=1}^q \weighting_i\expec{\pd{\obs,\pargen}}[(\pargenh_i(\obs) - \pargen_i)^2] 
%	\end{align*}
\end{itemize}
\end{proposition}
This result is proved in Appendix~\ref{sec:proof-all-previous-measurements}. Relations~\eqref{eq:updates-identities} are what we would expect without selection of sensing parameters using previous measurements. 
Next we derive the \gls{WWB} for a frequency estimation model based on array measurements.

%\newpage

\section{Statistical model for  multi-dimensional frequency estimation for random initial phase }\label{sec:wwb-and-models-random-phase-array-processing}

In this section we derive the statistical performance bound based on the \gls{WWB} for a family of array processing models under Gaussian and independent uniform priors. This metric can be used both for adaptation of transmission variables and for optimal design of constrained sparse arrays and sampling schemes (cf.~\cite{MGH-DMN-CG-RS:18}). 
%
%We emphasize Gaussian and independent uniform priors and point out considerations for solving the optimization over test points required for evaluation of the bound. 

\subsection{Observation model for spatio-temporal sampling}\label{sec:obs-model}

Here we present an observation model  for array processing tasks that include \glsreset{DoA}\gls{DoA} and Doppler estimation of a single source, and also \gls{MIMO} schemes such as \gls{TDM} \gls{MIMO}.

Consider the following data model for an observation $\obs$ according to a sensing scheme depicted in Fig.~\ref{fig:diagram-data-cube},
\begin{align}\label{eq:obs-model-array-processing}
\obs &= \sv(\pargen) \sqrt{\SNR} + \noise
\in \C^\numobs ,
\end{align}
%
%\begin{align}\label{eq:general-gaussian-model}
%\obsk = \gmodel(\pargenk, \sensingpolk) + \noise_k \in \C^\numobs
%\end{align}
%
where $\pargen \define (\freq_1, ..., \freq_{\numpar-1}, \phase)^T\in \R^{\numpar}$ is the vector of unknown parameters,  
$\noise\sim \Gaussian{\C}(\bm{0}, \eye{N})$ is standard complex Gaussian noise,
and $\SNR$ is the (single element) \gls{SNR}, which is assumed known or estimated beforehand. 
Furthermore $\sv(\pargen)$ denotes the \textit{spatio-temporal steering vector} for one source with frequencies $\freq_j$ and initial phase~$\phase$, defined as
\begin{align}\label{def:spatio-temporal-sv}
\sv(\pargen) &\define  e^{i \arrex \pargen} = e^{i\sum_j \arr{j} \freq_j}e^{i\phase}  \in \C^\numobs ,
\end{align}
which depends on the $\numpar-1$ \textit{sampling vectors} $\arr{j} \in \R^{N}$, which we combine, for convenience, 
to form the  \textit{sampling matrix},
\begin{align}\label{def:sampling-matrix}
\arrex &\define (\arr{1},..., \arr{\numpar-1}, \ones{\numobs})\in \R^{\numobs\times \numpar} .
\end{align}
 The sampling matrix $\arrex \equiv \arrex(\sensingpol)$ can be parametrized by a sensing parameter $\sensingpol$ that can be designed or adapted, and which we omit in this section. 
 We refer to the generic parameter $\Freq = (\freq_1, ..., \freq_{\numpar-1}$) as \textit{frequency} vector, whereas $\phase$ is called \textit{initial phase} or \textit{phase}.

This data model can be applied to the estimation of several quantities for one source, including \gls{DoA} or Doppler estimation, where $\arr{1}$ refers to antenna positions or pulse times; joint azimuth-elevation estimation with $3$-dimensional arrays~\cite[eq. (38)]{DTV-AR-RB-SM:14}, where $\arr{1}$, $\arr{2}$, $\arr{3}$, are the coordinates of the antenna locations in some basis and $\pargen$ are the electronic angles; and range-Doppler-azimuth estimation~\cite{FE-PH-AMZ-FKJ-MW:17} in automotive applications after a Fourier transform in the fast-time domain for each Tx and Rx pair and each pulse.

Next we show an example of application to the problem of \gls{TDM} \gls{MIMO} array processing for joint \gls{DoA}-Doppler estimation~\cite{KR-MV-BY:14}. This is a general template that we use in Section~\ref{sec:simulation-adaptive-channel-selection} for adaptive selection of antenna elements in the scenario of DoA estimation.

\subsection{\gls{TDM} \gls{MIMO} \gls{DoA}-Doppler estimation}

Adaptive sampling for \gls{DoA}-Doppler estimation can involve the activation of receivers and transmitter activation sequences.
Consider $\numRx$ receivers in a linear array at positions $\dRx\in \R^{\numRx}$ that collect the echoes from a total of $\numPulse$ pulses transmitted by a subset of a total of $\numTx$ transmitters available, allowing repetitions, located at positions $\dTx \in\R^{\numTx}$. The pulses are sent one after the other at time instances $\tPulse \in \R^{\numPulse}$. 
The transmission variables to be optimized are (i) the specific \textit{subset} and \textit{order} \textit{of transmitter activations}, codified by the matrix $\tsm\in \{0, 1\}^{\numPulse\times \numTx}$ (where $\ones{\numPulse} = \tsm\ones{\numTx}$);
%with  positions 
%$\dPulse = \tsm \dTx \in \R^{\numPulse}$;
(ii) the \textit{subset of} $\numRxsel$  \textit{receivers} whose signals are processed, codified by,
$\rsm\in \{0, 1\}^{\numRxsel\times \numRx}$; and (iii) possibly the carrier frequency. 
%with positions $\rsm\dRx$
A model for the observation of a single target with \gls{DoA} $\doa$, Doppler frequency $\dop=4\pi\velr/\lambda$, % $\dop=2\velr/\lambda$
and complex amplitude $\amplitude = \magnitude e^{i\phase}$, is given by \cite[eq. (4),(5)]{KR-MV-BY:14}
 \begin{align}\label{def:TDMmodel}
 \obs &= \svtdm(\doa, \velr,\phase) 
\magnitude  + \noise, 
 \end{align}
 where the spatio-temporal steering vector for \gls{TDM} \gls{MIMO},  %$b(\doa, \dop,\phase)$
 \begin{align*}
% \svtdm(\doa, \velr,\phase) \define e^{i \tfrac{\lambda}{\lambda_0}(\dVirt \doa + \tVirt 2\velr)}e^{i\phase} ,
 \svtdm(\doa, \velr,\phase) \define e^{i \tfrac{1}{\lambda}(\dVirt \doa + \tVirt \velr)}e^{i\phase} ,
 \end{align*}
can be written as in~\eqref{def:spatio-temporal-sv} 
%in distance units of $\tfrac{\lambda_0}{2\pi}$ and scaling $\tfrac{\lambda}{\lambda_0}$ of the nominal carrier frequency $f_0=\tfrac{c}{\lambda_0}$, 
in terms of the positions of the \textit{active} virtual elements, i.e., pairs Tx, Rx, and the pulse times,
%$\dVirt
%%\equiv \dVirt(\tsm,\rsm)
%\define \tsm\dTx \oplus \rsm\dRx 
%= \tsm\dTx \otimes \ones{\numRxsel} + \ones{\numPulse}\otimes \rsm\dRx 
%$; and $
%\tVirt \define  \tPulse \otimes \ones{\numRxsel}$,
\begin{align*}%\label{def:TDMmodel-aux-equations}
\dVirt
%\equiv \dVirt(\tsm,\rsm)
%
&\define \tsm\dTx \otimes \ones{\numRxsel} + \ones{\numPulse}\otimes \rsm\dRx 
% &\define \tsm\dTx \oplus \rsm\dRx 
% \\
%&= \tsm\dTx \otimes \ones{\numRxsel} + \ones{\numPulse}\otimes \rsm\dRx 
\\
\tVirt &\define  \tPulse \otimes \ones{\numRxsel},
\end{align*}
(units of $\frac{1}{2\pi}$ for $\dVirt$ and $\frac{1}{4\pi}$ for $\tVirt$)
yielding the sampling matrix
\begin{align*}
\arrex(\sensingpol)=\tfrac{1}{\lambda}(\dVirt, \tVirt, \lambda\ones{\numPulse\numRxsel} )\in\R^{\numPulse\numRxsel \times 3} .
%\arrex(\sensingpol)=\tfrac{\lambda}{\lambda_0}(\dVirt, \tVirt, \ones{\numPulse\numRxsel} )\in\R^{\numPulse\numRxsel \times 3} .
\end{align*}
This fits into our general model \eqref{eq:obs-model-array-processing} by identifying  $\pargen = (\doa, \velr, \phase)$ as parameters to be estimated, $\SNR = \frac{\magnitude^2}{\noisevar}$ as \gls{SNR},  $\arr{1} = \frac{\dVirt}{\lambda}$ as virtual array positions, and $\arr{2} = \frac{\tVirt}{\lambda}$ as virtual pulse times. 

In the next section we describe the construction of the WWB for the model~\eqref{def:spatio-temporal-sv}, \eqref{def:sampling-matrix} that includes the above scenarios.

\begin{figure}[]
	%\hspace*{-0.75cm}
	%	\subfigure[Title 2]
%	{\includegraphics[width=1\linewidth]{figures/generalfigures/diagram_data_cube2.jpg}}
		{\includegraphics[width=1\linewidth]{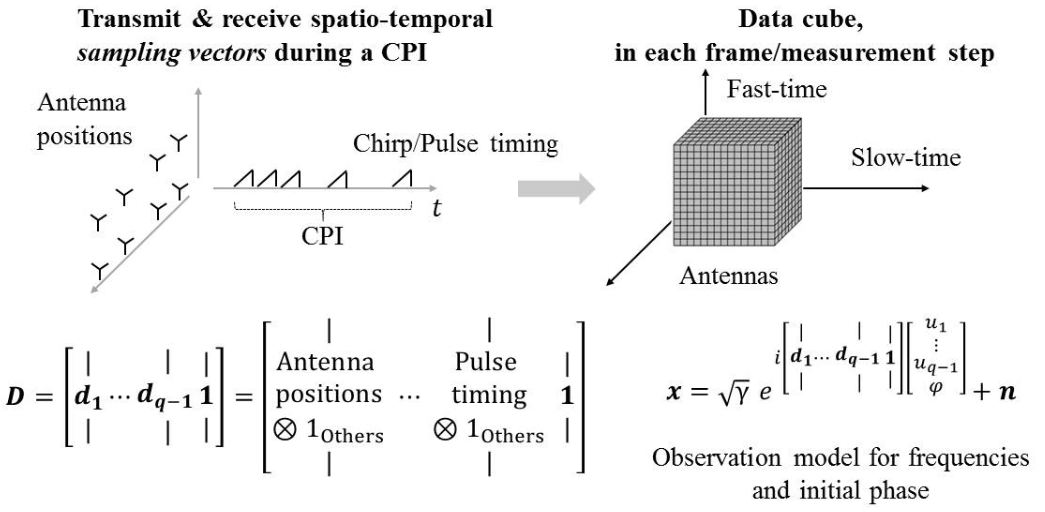}}
	%	\subfigure[Title 2]
	%	{\includegraphics[width=.48\linewidth]{}
	%	\begin{picture}(0,0)(0,0)
	%	%\put(-93,42){$$}  
	%	\end{picture}
	\caption{Design of sensing parameters in spatial and temporal domains during a coherent processing interval (CPI), or frame, with random initial phase. The antenna positions and pulse timing can be nonuniform and sparse and can be selected or scaled between frames using the proposed Bayesian adaptive framework.}\label{fig:diagram-data-cube}
\end{figure}

\subsection{Random-phase WWB for array processing}\label{sec:wwb-random-phase}
Here we derive the \gls{WWB} %(cf. Section \ref{sec:WWB-theory} in the Appendix)
 for the data model introduced in~\eqref{eq:obs-model-array-processing} both for Gaussian and independent uniform belief distributions on the frequency parameters. The calculation is similar to the general formulation in \cite{DTV-AR-RB-SM:14}, but employs a different class of test points.
\\
Performing the calculations outlined in the Appendix \ref{sec:WWB-theory} 
%(equations \eqref{eq:def-wwb-parametric},\eqref{eq:def-wwb-Q}, \eqref{eq:eta-general} and \eqref{eq:def-etaacute-gaussian})
for the choice of test point $ \tphI = (\tphI[u_1], ..., \tphI[u_{q-1}], \tphI[\phase])^T\in \R^{\numpar\times 1},$ we obtain
%\begin{align}\label{eq:tpOneColumn}
%\tphI = (\tphI[u_1], ..., \tphI[u_{q-1}], \tphI[\phase])^T\in \R^{\numpar\times 1},
%\end{align} 
\begin{subequations}\label{eq:randomphase-WWB-general}
\begin{align}
\wwb(\tphI) &= \frac{1}{\QI}\tphI \tphI^T \in \C^{\numpar\times \numpar}
\\
\QI &= 2 \frac{\eta(\tphI,\tphI)-\eta(\tphI,-\tphI)}{\eta(\tphI,\bm{0})^2}
\\
\eta(\tpgenI, \tpgenII) &= \acute{\eta}(\tpgenI, \tpgenII) \intprior(\tpgenI, \tpgenII) \label{eq:eta-factorized}
\end{align}
\end{subequations}
where the integral over observations $\acute{\eta}\equiv \acute{\eta}_{\pargen}$ according to~\eqref{eq:def-etaacute-gaussian}~is
\begin{align}\label{eq:rpwwb_etaacute}
\acute{\eta}(\tpgenI, \tpgenII) &=  \exp(-\frac{\SNR}{2}(\numobs-\repart{ \ones{\numobs}^T e^{i \arrex(\tpgenII-\tpgenI)}} )) ,
% & = \exp(-\frac{\SNR}{2}(\numobs-\repart{ \sum_{n = 1}^\numobs e^{i(\tpgenII_{\phase}-\tpgenI_{\phase})}e^{i\sum_j d_{j,n}(\tpgenII_{\freq_j}-\tpgenI_{\freq_j})}} ))  
\end{align}
%by means of some simple algebraic manipulation 
and the following integral over parameters remains to be determined after the choice of prior distribution,
\begin{align}\label{eq:intpriorsGen}
\intprior(\tpgenI, \tpgenII) &\define \int_{\Pargen}\pd{\pargen} \sqrt{\frac{\pd{\pargen+\tpgenI}\pd{\pargen+\tpgenII}}{\pd{\pargen}^2}} d\pargen .
\end{align}
Analytic expressions for this integral are given in the Appendix \ref{sec:IntegralOverPriors} %$\intprior(\tpgenI, \tpgenII)$ 
  in the cases of Gaussian \eqref{intpriorsGaussian} and independent uniform \eqref{intpriorsUniform} priors.

\begin{remark}[Choice of test point matrix]\label{rem:WWBtpcolumn}
The computation of the \gls{WWB} \eqref{eq:def-wwb-parametric} requires to select a test point matrix $\tpH$. %according to~\eqref{eq:def-tpH}. 
In \cite[sec. 4.4.1.4]{HLVT-KLB:13}, it is suggested to use test point matrices with at least as many columns $M$ as rows $\numpar$ (i.e. number of random parameters in the model). While still being valid lower bounds, the WWB matrices that arise from test points with $M<q$ are by construction rank-deficient and therefore suboptimally suited to produce tight bounds to a presumably full-rank \gls{BMSE} matrix.
%On the other hand, we are not primarily interested in actual MSE values, but rather in the relative benefits of one transmission variable over the other in terms of MSE for the frequency parameter $\Freq$.\\
For the sake of simplicity in terms of derivation and computational time, we nonetheless select the test point matrix %in~\eqref{eq:tpOneColumn}
 to comprise only one column and perform global optimization 
%over this kind of test points 
to find the tightest bound in this class.
\end{remark}

\begin{remark}[Factorization of integrals]
Note that the factorization of $\eta = \acute{\eta} \cdot \intprior$ in~\eqref{eq:eta-factorized} (cf. equation~\eqref{eq:eta-general} in the Appendix) into the product of an integral over observations and an integral over parameters is due to the fact that, for this choice of model and parameters, 
the function $\acute{\eta}$ is independent of $\pargen$.
 Unfortunately, this is not true if e.g. the \gls{SNR} $\SNR$ is included into the set of random parameters $\pargen$ or in the case of several targets. The consequence of the latter is an increase in computation time.
\end{remark}

Combining the expressions in~\eqref{eq:randomphase-WWB-general},
the optimization problem~\eqref{def:costwwb} can be written as
\begin{align}\label{eq:rpwwbopt}
\cost &= \sup_{\tphI\in \tpDom} \traceW\wwb(\tphI)   
% % Note: when refering to the objective function, we have it above; when referring to the WWB, we can use the equations spelled at the beginning of the section becasue the next is almost redundant.
%\\
%\label{eq:rpwwb}
%\wwb(\tphI) &= \frac{\tphI \tphI^T}{2}  \frac{\acute{\eta}(\tphI, \bm{0})^2 \intprior(\tphI, \bm{0})^2}{\intprior(\tphI, \tphI) - \acute{\eta}(\tphI, -\tphI)\intprior(\tphI, -\tphI)}
\end{align}
for $\tphI$ in the domain (cf.  \eqref{eq:domain-condition} in the Appendix)
	\begin{align}\label{eq:domain-conditionI}
	\tpDom\define \setdef{ \tphI\in\R^{\numpar\times 1} }
	 {\Pargen\cap( \Pargen + \tphI) \neq \emptyset }.
	\end{align}
	Note that this set depends on the prior belief through its support, $\Pargen \define \operatorname{supp}(\pd{\pargen})=\setdef{\pargen\in\R}{\pd{\pargen}>0}$.
%	\begin{align*}
%	\Pargen \define \operatorname{supp}(\pd{\pargen})=\setdef{\pargen\in\R}{\pd{\pargen}>0} .
%	\end{align*}

\begin{remark}[Symmetry of \gls{WWB} and $\tpDom$]\label{rem:symmetryWWB}
A sufficient condition for the symmetry $\wwb(\tphI) = \wwb(-\tphI)$
%	\begin{align*}
%\wwb(\tphI) = \wwb(-\tphI)
%	\end{align*} 
results from the corresponding symmetries $(\tpgenI, \tpgenII) \to (-\tpgenI, -\tpgenII)$ 
for $\acute{\eta}$, which holds in general in view of~\eqref{eq:rpwwb_etaacute}, and also
for $\intprior$, which depends on the specific prior distribution.
Since $\tphI\in\tpDom \iff -\tphI\in\tpDom$
%\begin{align*}
%\tphI\in\tpDom \iff -\tphI\in\tpDom
%\end{align*}
 by definition of $\tpDom$ in~\eqref{eq:domain-conditionI},
  this symmetry implies that we can neglect those test points of $\tpDom$ which are on one side of an arbitrarily chosen hyperplane through $0\in \tpDom$ when solving the optimization problem~\eqref{eq:rpwwbopt}. We thus choose to restrict the optimization to test points with positive phase component $\tphI_{\phase}\geq 0$.
%  , which seems the natural candidate for our model.
\end{remark}

\cg{
The formulas %\eqref{eq:rpwwb} 
\eqref{eq:rpwwbUniform}, \eqref{eq:rpwwbGaussian} for the \gls{WWB} with uniform and Gaussian priors have been implemented in Matlab as vectorized functions, which makes the optimization over test points $\tphI$ \eqref{eq:rpwwbopt} more endurable. (could share code)\\
As optimization method we extended the simulated annealing algorithm offered by \ref{} to handle more than one 'particle' in parallel and added a stopping criterion based on average relative improvement compared to function values in previous temperatures.
(one optimization with stopping criterion takes less than a second for the array scaling example below).\\
}

For both uniform and Gaussian priors, the integral over priors $\intprior$ and thus the corresponding \gls{WWB}~\eqref{eq:rpwwbUniform}, \eqref{eq:rpwwbGaussian} depend only on the variance
%$\Freqvar$ 
%(which is analogous to the length of the support $\Delta \Freq$ for uniform priors) 
of the prior, and is independent of its mean.\\% $\Freqmean$.\\
The dependence on the prior only through the variance makes it convenient to numerically characterize the controller based on the corresponding \gls{WWB} cost function for observation models in low-dimensional estimation and for low-dimensional sensing parameters. We study the controller output based on inputs given by distributions with this property for an array scaling task in the following section.

\section{Analysis of random-phase \gls{WWB} for scaling of sampling period}\label{sec:controller-wwb}

Here we analyze the \glslink{WWB}{WWB} constructed in Section~\ref{sec:wwb-and-models-random-phase-array-processing} for a model of frequency estimation in one dimension plus an extra dimension for the initial phase. We compare them with the corresponding bounds assuming a known initial phase in the decision problem of selecting the optimal scaling of the array. In particular, we characterize the controller choices numerically using a look-up table that can be used for real-time computation and also interpret the choices, specially the added robustness of considering the initial phase unknown.
%trade-off between ambiguity and aliasing suppression and resolution in a case study of uniform array scaling.

\subsection{Observation model for array scaling}
We investigate the following estimation problem as a special case of model~\eqref{eq:obs-model-array-processing}: Consider a sampling vector $\arrI\in \R^{\numobs}$ that can be scaled by a factor $\scaling_k\in\R_{+}$, that we wish to adapt at each measurement step  $k = 1, 2, \dots$, according to the data model 
\begin{align}\label{eq:ArrayScalingMeasurement}
\obs_k &= e^{i\scaling_k \arrI \freq}e^{i\phase_k}\sqrt{\SNR} + \noise_{k}.
\end{align}
We recall that the Gaussian noise realizations $\noise_{k}\sim \Gaussian{\C}(0,\eye{\numobs})$ and the uniform random initial phases $\phase_k\sim \Uniform{[-\pi, \pi]}$ are assumed independent from each other and between time steps.

%\begin{remark}[Discussion of applications of adaptive array scaling]
Regarding applications, the sampling vector and the frequency parameter in model~\eqref{eq:ArrayScalingMeasurement} can have various interpretations. For instance, the scaling parameter can be the (inverse) wavelength $\lambda$ of the carrier frequency in narrow-band \gls{DoA} estimation, or the \gls{PRI} for a train of pulses in Doppler estimation (after the range bin is determined). For adaptation steps on a time scale exceeding the coherency interval of the radar, the prescribed randomization of the initial phase is necessary. 
It is worth noting that even though these applications can be cast into the mathematical shape of model \eqref{eq:ArrayScalingMeasurement}, it is necessary to reflect on how changing the scaling affects other radar properties in the context of the broader estimation task (e.g. changing the \gls{PRF} affects range ambiguities). On the other hand, the more general array processing model~\eqref{eq:obs-model-array-processing} allows to consider these broader scenarios.\\
	
%\end{remark}

\begin{figure}[]
	\begin{minipage}{1.00\columnwidth} 
		
		%   \centering
				\includegraphics[width = 1.00\columnwidth]{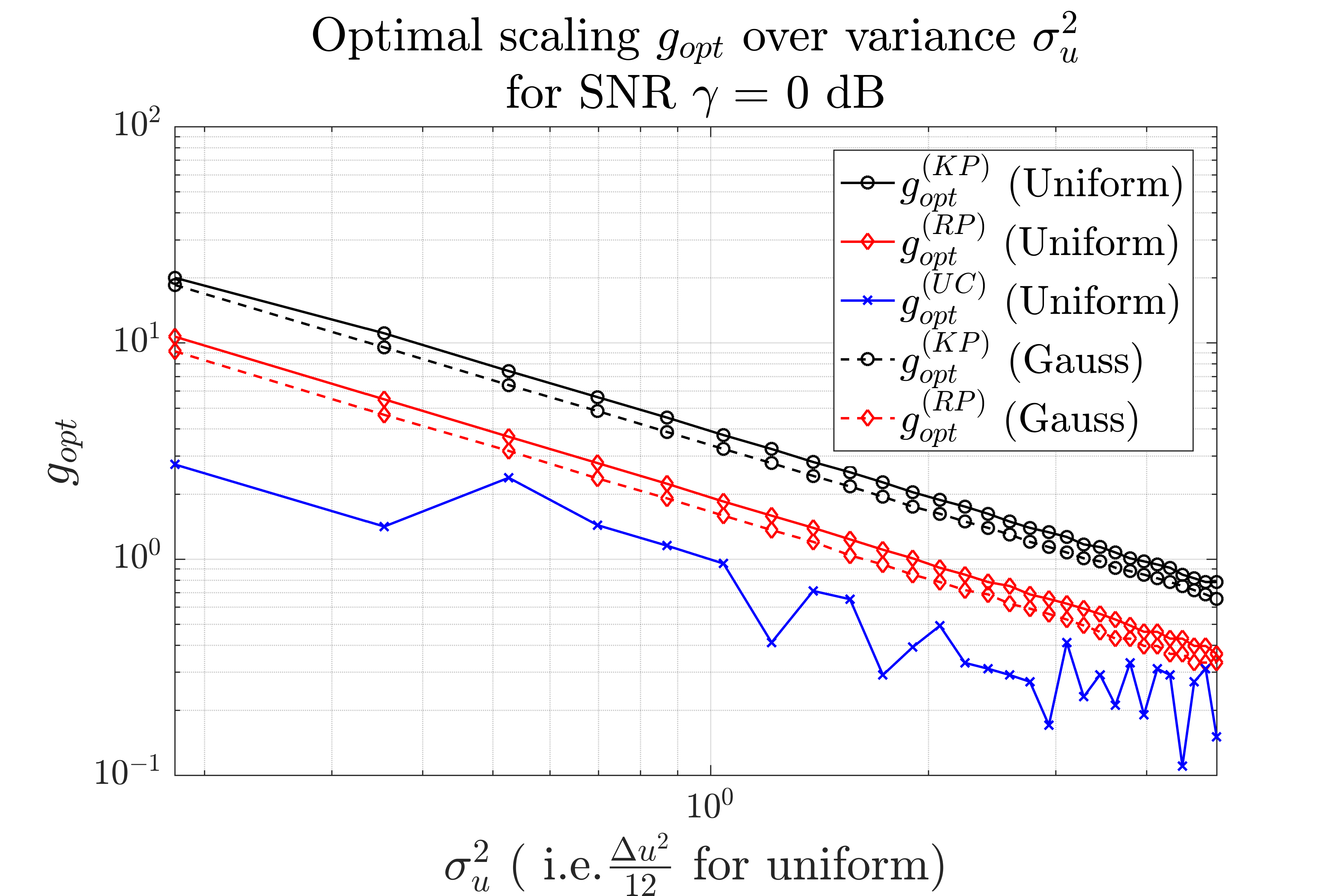}
		\caption{Optimal scaling versus variance $\sigma^2_\freq$ of Gaussian and uniform priors. The optimal choices according to the Random-Phase \gls{WWB} are depicted in red, the choices for the Known-Phase \gls{WWB}~\cite[Eq. (57)]{DTV-AR-RB-SM:14} in black, and for the unconditional \gls{WWB}~\cite[Eq. (56)]{DTV-AR-RB-SM:14} in blue.}
		\label{figure:GAllControllerSNRhigh}
		
	\end{minipage}
	%As $\scaling = 15$ was the highest grid value considered for the scaling parameter, we observe a capping effect for $\Delta \freq \to 0$.
\end{figure}

\begin{figure}[]
	\begin{minipage}{1.00\columnwidth} 
		
		%   \centering
		
		\includegraphics[width = 1.00\columnwidth]{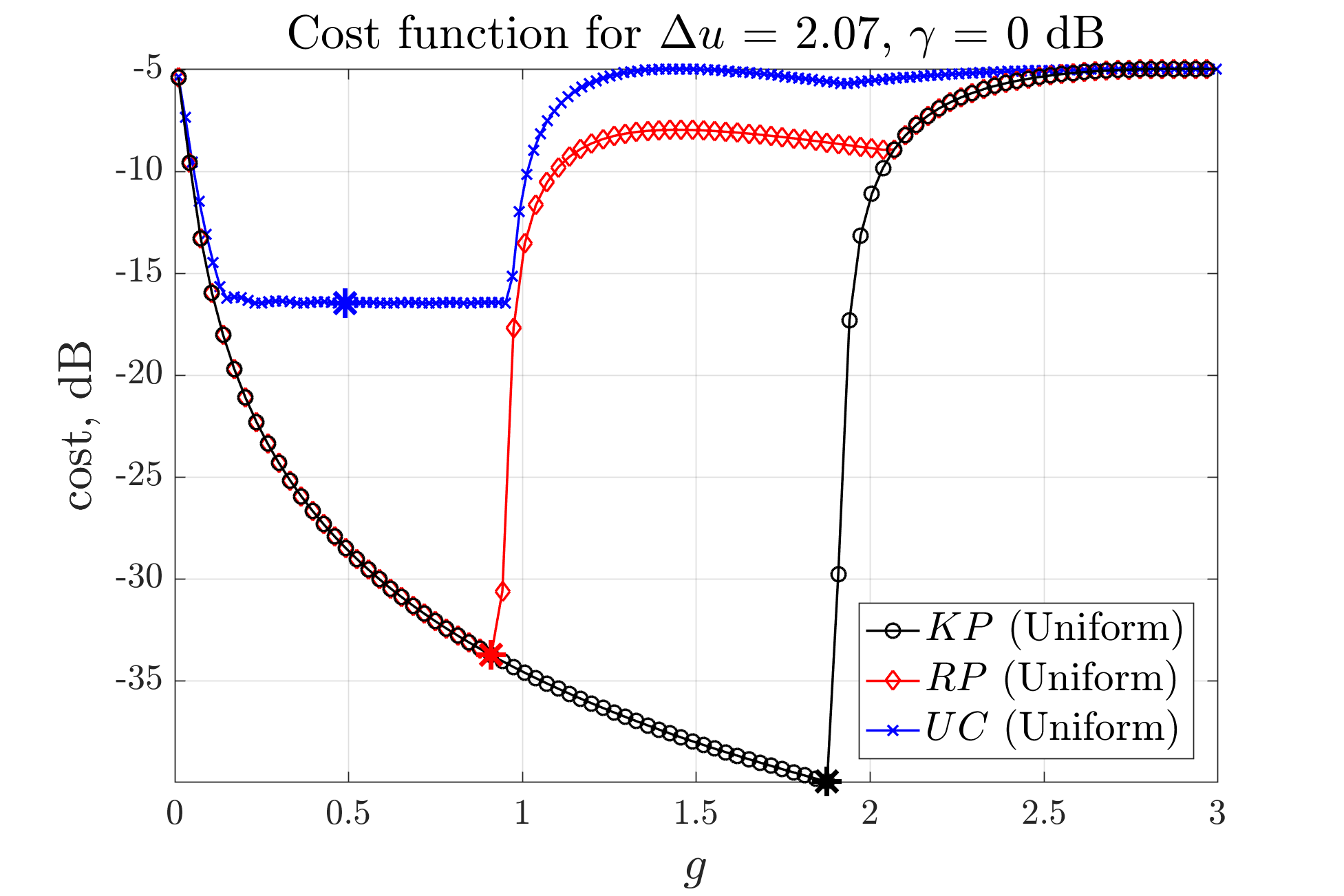}
		\caption{Cost function $\cost(\scaling)$ for all three models. The unconditional model
			optimal scaling is harder to determine precisely for the unconditional model due to the slight variations in the cost function `basin' and the optimal value is smaller and thus more conservative.}
		\label{figure:BasinCostUniform}

	\end{minipage}
	
\end{figure}

\subsection{Characterization of controller for uniform array scaling}\label{sec:characterization-controller}

Here we analyze the controller choices based on optimization of the \gls{WWB} cost function \eqref{eq:rpwwbopt} for uniform and Gaussian priors, respectively given by expressions \eqref{eq:rpwwbGaussian} and \eqref{eq:rpwwbUniform} in Section \ref{sec:wwb-random-phase}. To this end, we visualize the dependence of the optimal scaling on prior variance and \gls{SNR} value and compare with the results obtained from related statistical performance bounds.

The scaling choices are computed for a uniform array of $\numobs = 12$ elements with positions $\arrI =  \pi (1,\dots,\numobs)-\tfrac{\numobs+1}{2}$
%\begin{align}\label{eq:uniformArray}
%\arrI = (\arrI_n)_{n = 1}^\numobs = \pi (n-\frac{\numobs+1}{2} )_{n = 1}^\numobs .
%%\arrI = \pi (0,..., \numobs-1)^T - \pi \frac{\numobs-1}{2} 
%\end{align}
(naturally, other arrays are possible). The array's center of mass is at the origin, i.e. $\arrICM = \frac{1}{\numobs}\sum_n \arrI_n = 0$. % and SNR $\SNR = 1$ (i.e. 0 dB). 
We denote the optimal scaling choice according to the presented \gls{RP} \gls{WWB} in \eqref{eq:rpwwbUniform} and \eqref{eq:rpwwbGaussian},  as $\scaling^{(RP)}_{opt} \equiv \scaling^{(RP)}_{opt}(\sigma^2_{\freq})$. Those depend only on the variance $\sigma^2_{\freq}$ of the Gaussian or uniform belief distribution, which indicates the certainty we have on the parameter $\freq$. %This dependence on the distribution only through the variance is a property exhibited by Gaussian and uniform priors.\\
% (Note that for uniform and Gaussian priors, the WWB only depends on the variance of the distribution.) 
%
For comparison, we discuss also %in~Fig. \ref{figure:GAllControllerSNRhigh} 
the optimal scaling choices according to the \glspl{WWB} for two slightly different models: the scalings referred to as $\scaling^{(KP)}_{opt}$ are found from a \textit{\gls{KP}} model which assumes the phase $\phase$ as known~(based on \cite[Eq. (57)]{DTV-AR-RB-SM:14}), while the scalings referred to as $\scaling^{(UC)}_{opt}$ correspond to an \textit{\gls{UC}} model with Gaussian amplitude based on \cite[eq. (56)]{DTV-AR-RB-SM:14}. Details on the latter are provided in the Appendix \ref{sec:UnconditionalWWB}.\\

We focus on two aspects of the scaling selection, (i) the dependence on variance and \gls{SNR}, and (ii) the sensitivity of the \gls{WWB} with respect to variations of the optimal scaling.

\paragraph{Dependencies of \gls{RP} and \gls{KP} models}
%(i) 
Regarding the dependence on variance (or \glslink{FoV}{Field-of-View} length), consider
Fig.~\ref{figure:GAllControllerSNRhigh} that shows the scaling selections for the \gls{RP} model with uniform and Gaussian priors (red). As we expect, the more uncertainty we have about $\freq$ (i.e. $\sigma^2_\freq$ large), the smaller the optimal  array scaling $\scaling$ to avoid aliasing, and vice-versa, the more certainty about $\freq$ (i.e. $\sigma^2_\freq$ small), the larger the scaling, to trade off ambiguity suppression with accuracy. 
The difference between Gaussian and uniform priors of same variance is comparatively insignificant. 
We observe that in this \textit{high \gls{SNR} scenario} with $\SNR = 0$ dB, the choices based on the \gls{KP} \gls{WWB} (black) are roughly twice the value as compared to the random-phase model, $\scaling^{(KP)}_{opt} \approx 2\scaling^{(RP)}_{opt}$. Further inspection (not shown here), reveals that this relationship holds for \gls{SNR} values approximately above $\SNR = -1.5$~dB.
For lower \gls{SNR} values, we find that the optimal scalings according to the random and known-phase models coincide $\scaling^{(KP)}_{opt} \approx \scaling^{(RP)}_{opt}$.

\begin{remark}[Dependence of \gls{KP} \gls{WWB} on coordinate origin]\label{rem:dependence-coordinate-origin}
%But these relationships for the
The known-phase \gls{WWB} depends on the coordinate origin chosen to define the sampling vector $\arrI$, which in these visualizations is taken as the array center of mass. This makes it harder to analyze this WWB model and also argues against using it.
This dependence also occurs for the \gls{CRB} where the Fisher information matrix satisfies $\FIM =2 \SNR \scaling^2 \Vert \arrI\Vert^2$ for the known-phase model, while for the random-phase model  $\FIM = 2 \SNR \scaling^2 \Vert \arrI -\arrICM\ones{\numobs}\Vert^2$.
\end{remark}
 
 \paragraph{Dependencies of \gls{UC} model}
The unconditional model considers a signal with random amplitude $\amplitude\sim \Gaussian{\C}(0, \SNRUC)$ (cf. Appendix~\ref{sec:UnconditionalWWB}) yielding a notion of SNR $\SNRUC$ with a different interpretation. The deterministic SNR notion $\SNR$ of the other models obey an exponential distribution $\SNR\sim  \mathrm{Exp}(\SNRUC)$ under this model,
%and $\SNRUC=\expec{}[\SNR^2]$, 
such that, even though $\expec{}[\SNR] = \SNRUC$, a given value of $\SNRUC$ emphasizes low SNR values according to the previous notion.
The optimal scaling $\scaling^{(UC)}_{opt}$ (blue) in Fig. ~\ref{figure:GAllControllerSNRhigh} (for $\SNRUC = 0$ dB) is thus more conservative compared to the other models for high SNR. 
The unsteady behavior of the curve for the unconditional model can be understood in view of Fig. \ref{figure:BasinCostUniform}, which shows the cost function plotted over scalings. We observe that the unconditional model exhibits an almost flat 'basin' of low values with only slight oscillation in which the optimal scaling lies for this model. It is thus likely that the scaling in the optimization grid with the lowest cost function value is found in a neighboring peak of the analytical optimum. 
 At low SNR, the respective scalings of the other two models (not shown here) are smaller than $\scaling^{(UC)}_{opt}$. In summary, the optimal scaling does not depend very strongly on the SNR $\SNRUC$ for the unconditional model and is generally more conservative.

\begin{figure}[]
	\begin{minipage}{1.00\columnwidth} 
		
		\centering
		
		\includegraphics[width = 1.00\columnwidth]{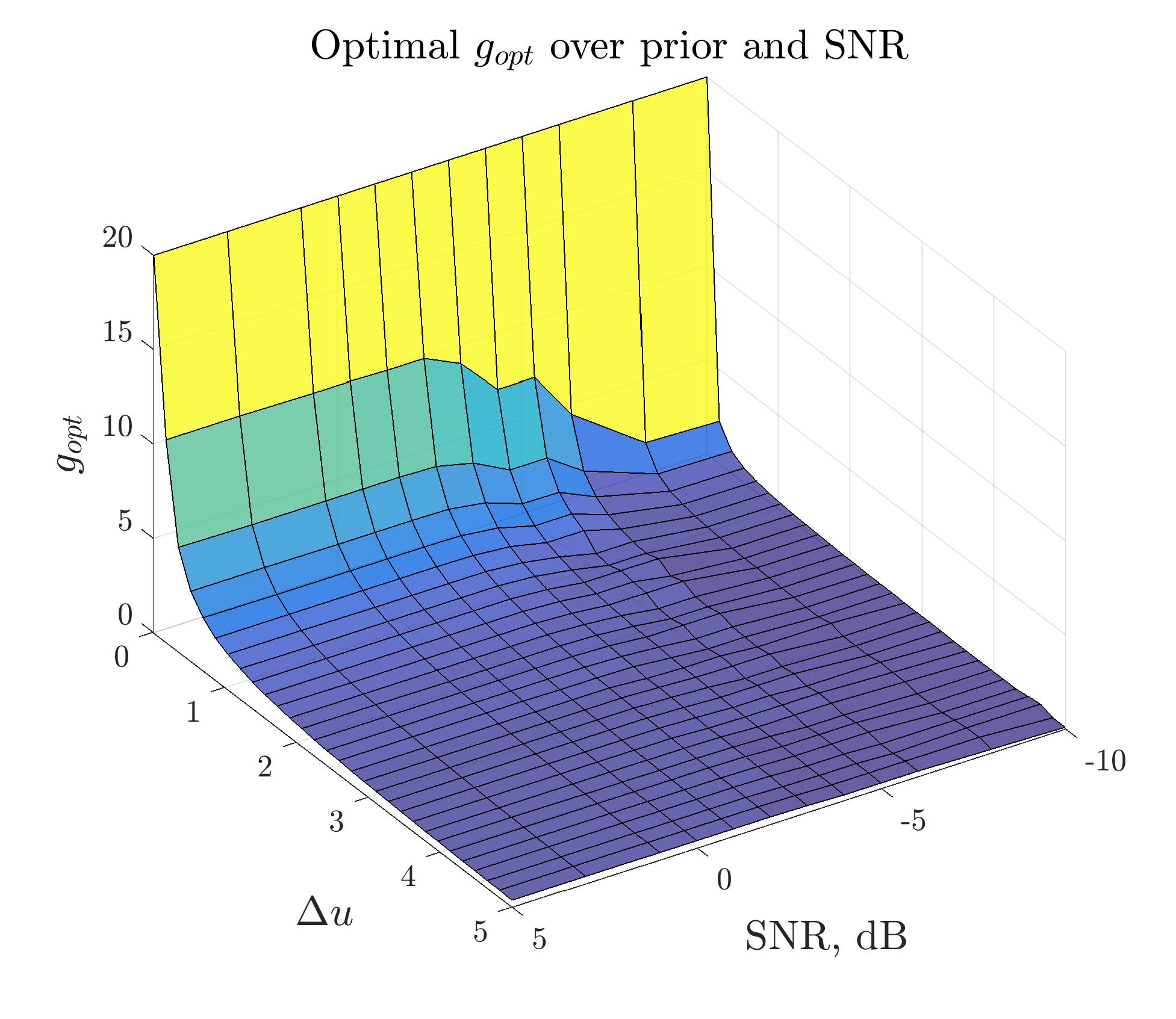}
		\caption{Optimal scaling according to \glslink{RP}{random-phase} \gls{WWB} versus \gls{SNR} $\SNR$ and support length $\Delta \freq$ of uniform prior.
			% Vertical slices do not change for very high SNR or very low SNR. For low values of SNR, curves do not scale like $\frac{1}{\Delta \freq}$, but with a higher power.
		}
		\label{figure:GsurfaceSNRpriorRPWWB}
		
	\end{minipage}
\end{figure}

\paragraph{Sensitivity of \gls{RP}  model}
%(ii)
 With regards to the sensitivity of the \gls{RP} \gls{WWB} with respect to scaling, consider
%More details for uniform priors with support length $\Delta \freq = \sqrt{12\sigma_\freq^2}$:\\
Figure~\ref{figure:costsurfaceMerge} depicting the optimal choices for the case of uniform priors at both high and low \gls{SNR}, including the value of the \gls{RP} \gls{WWB} cost (in color code) for each pair $(\Delta\freq,\scaling)$.
It is noteworthy, from a sensitivity perspective, the change of the \gls{RP} \gls{WWB} with respect to scaling choice.
We observe that for high \gls{SNR} (cf. left plot in Fig.~\ref{figure:costsurfaceMerge}), the optimal scaling according to the \gls{RP} \gls{WWB}  $\scaling^{(RP)}_{opt}$ (red), is only slightly smaller than scalings which abruptly exhibit significantly higher cost values, whereas for low \gls{SNR} (cf. right plot of Fig.~\ref{figure:costsurfaceMerge}), the optimal scalings are not so close to such a \textit{threshold}. 
This is relevant because for high \gls{SNR} the optimal choices for the alternative metric of \glslink{KP}{known-phase} \gls{WWB} $\scaling^{(KP)}_{opt}$ (black) are slightly bigger and thus in the region of higher cost from the perspective of the \gls{RP} \gls{WWB}. This phenomenon is studied in the next section, where we show using the array factor that the \gls{RP} \gls{WWB} captures the notion of aliasing differently than the \gls{KP} \gls{WWB}.

\begin{figure}[]  % 
	\begin{minipage}{1.00\columnwidth} 
		%
		%  \centering
		
	%	\includegraphics[width = 1.00\columnwidth]{figures/wwb/costsurfaceMerge.png}
			\includegraphics[width = 1.00\columnwidth]{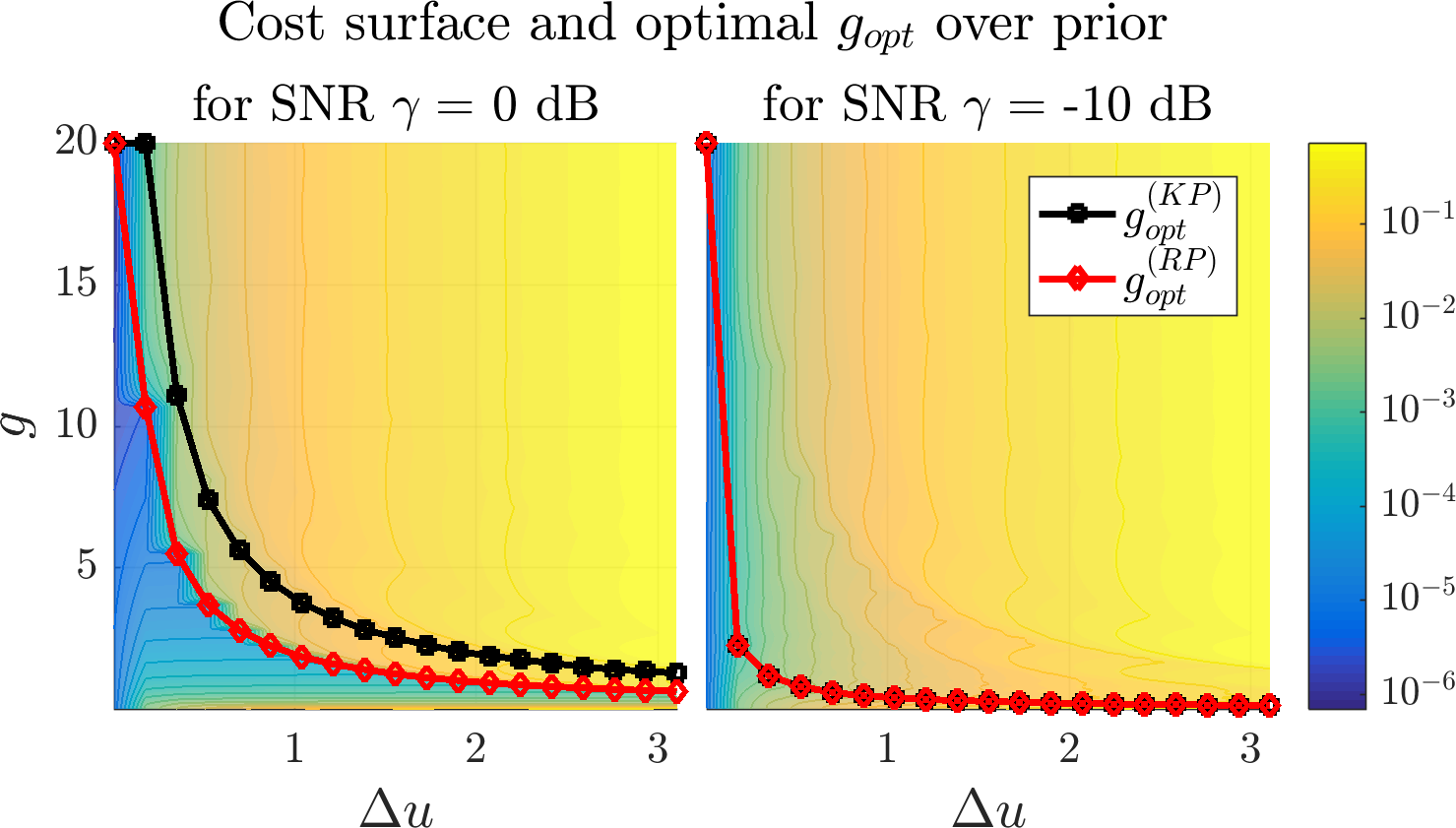}
		\caption{
\glslink{RP}{Color code: random-phase (RP)} \gls{WWB} with uniform prior for each pair $(\Delta\freq,\scaling)$.
			The red curve shows the optimal scaling for each $\Delta \freq$ (i.e. minimum of \gls{RP} \gls{WWB} along each vertical line) at the given \gls{SNR}. The black curve shows the corresponding scaling choice according to the \glslink{KP}{known-phase} \gls{WWB}. For high SNR (left), we observe the latter yields consistently larger scalings.
			 For low SNR (right), the choices of the \glslink{KP}{known-phase} \gls{WWB}, which depend on the convention for the choice of coordinate origin, are the same as for the \glslink{RP}{random-phase} \gls{WWB}. %Same as Fig.~\ref{figure:costsurfaceSNRhigh} for lower \gls{SNR}. In this case the choices of the \glslink{KP}{known-phase} \gls{WWB}, which depend on the convention for the choice of coordinate origin, are the same as for the \glslink{RP}{random-phase} \gls{WWB}.
		}
		\label{figure:costsurfaceMerge}
		
	\end{minipage}
	
\end{figure}

\subsection{Ambiguity function for optimal choices of known-phase and random-phase WWB for array scaling}

Here we interpret the behavior of the \gls{RP} \gls{WWB} and the \gls{KP} \gls{WWB} calculated in the previous section in terms of the array factor for the same model of frequency estimation~\eqref{eq:ArrayScalingMeasurement} as a function of array scaling.
The array factor or ambiguity function with respect to the sampling vector (or array) $\arrI$ is given by
\begin{align}\label{def:Beampattern}
B(\freq,\freq+\BPh) &\define 
\frac{\langle e^{i \scaling \arrI \freq}, e^{i \scaling \arrI (\freq+\BPh)}\rangle}
{\Vert e^{i \scaling \arrI \freq}\Vert\Vert e^{i \scaling \arrI (\freq+\BPh)}\Vert } \notag
%=
% \frac{\langle \ones{\numobs}, e^{i \scaling \arrI \BPh}\rangle}
% {\Vert \ones{N}\Vert\Vert e^{i \scaling \arrI \BPh}\Vert }
 \\
& =
 \frac{\langle \ones{\numobs}, e^{i \scaling \arrI \BPh}\rangle}
 {\numobs} 
 = \frac{1}{\numobs} \sum_{n = 1}^\numobs e^{i\scaling \arrI[n] \BPh}\equiv B(\BPh) .
%
%B(\BPh) = \frac{\langle \ones{\numobs}, e^{i \scaling \arrI \BPh}\rangle}{\Vert \ones{N}\Vert\Vert e^{i \scaling \arrI \BPh}\Vert } = \frac{1}{\numobs} \sum_{n = 1}^\numobs e^{i\scaling \arrI[n] \BPh}.
\end{align}
This quantity appears in the \gls{WWB} through the function $\acute{\eta}$ in~\eqref{eq:rpwwb_etaacute}. 
It can be interpreted in several ways: (i) Cosine distance or ambiguity function between signals 
$e^{i \scaling \arrI \freq}$; (ii) The \gls{DTFT}, or projection, of a signal $e^{i \scaling \arrI \freq}$ into the frequency-shifted signal $e^{i \scaling \arrI (\freq+\BPh)}$.
Figs.~\ref{figure:optimalBeampatternSNRhigh} and~\ref{figure:optimalBeampatternSNRlow} depict the array factor of the optimal arrays for the RP and KP models, respectively, for a specific \glslink{FoV}{Field-of-View} length~$\Delta u$ for two \gls{SNR} values, $\SNR=0$, regarded here as \textit{high}, and $\SNR=-10$, regarded as \textit{low}. \\
 We make the following observations:\\
 
(i) For \textit{high} \gls{SNR}, the \glslink{RP}{random-phase} \gls{WWB} favors the largest scaling which places the first grating lobe (i.e. smallest $\BPh_1>0$ with $\vert B(\BPh)\vert = 1$) right outside of $\Delta \freq$ as measured from the main lobe (c.f. Fig.~\ref{figure:optimalBeampatternSNRhigh}). This choice maximizes the accuracy (since larger apertures correspond to thinner mainlobes) \cg{thinnest main lobe to the bigger aperture?,} while avoiding aliasing even for the extreme case of the parameter value being in the extremes of the prior belief distribution.
This choice explains the sensitivity phenomenon displayed in Figs.~\ref{figure:BasinCostUniform} and \ref{figure:costsurfaceMerge} (left) where the cost function shows a dramatic increase for \textit{bigger} scalings. A lesson from this regarding the application of the \gls{RP} \gls{WWB} for adaptive scaling is the following: if the prior is given by a uniform approximation of the filter's empirical density output, then the support length $\Delta \freq $ needs to be chosen \textit{conservatively}, at least in the case of \textit{high} \gls{SNR}.\\

(ii) For \textit{low} \gls{SNR}, the behavior is governed by the sidelobes (c.f. Fig.~\ref{figure:optimalBeampatternSNRlow}).
In the context of Fig.~\ref{figure:costsurfaceMerge} (right), we note that the optimal scaling is in a region of relatively small slope or change of the \gls{RP} \gls{WWB}. It is noteworthy that for \textit{low} \gls{SNR} below approximately $ -1.5$ dB the \gls{KP} \gls{WWB} and \gls{RP} \gls{WWB} yield the same choice.\\

(iii) The \gls{KP} \gls{WWB} can identify only aliasing problems for test points that satisfy $\repart{B(\BPh)} = 1$.
Their location depends on the convention for the array's center of mass $\arrICM$. For a uniform array, 
%$\arrI =  \pi (1,\dots,\numobs)-\tfrac{\numobs+1}{2}$
%of~\eqref{eq:uniformArray}, 
the array factor is $B(\BPh) =\tfrac{1}{\numobs} e^{i\scaling \BPh \arrICM}\sin(\numobs\frac{\pi \scaling \BPh}{2})/\sin(\frac{\pi \scaling \BPh}{2}) .$
Choosing the array's center of mass as coordinate origin (i.e. $\arrICM = 0$), we find that every other of the grating lobes at $\BPh_k = \frac{2k}{\scaling}$ is a minimum for $\repart{B(\BPh)}$ if the number of observations $\numobs$ is even, and not a maximum as for $\vert B(\BPh)\vert$. 
This is especially the case for the first grating lobe $\BPh_1$ and thus the controller based on the \glslink{KP}{known-phase} \gls{WWB} chooses, for fixed $\Delta \freq$, a scaling \textit{twice as large than it should.}
\cg{(c.f. Fig.~\ref{figure:optimalBeampatternWWBold} and~\ref{figure:GsurfaceSNRpriorDiff}.)}
(In the case of the uniform array, we could choose an offset $\arrICM = \pm\frac{\pi}{2}$ to detect the grating lobe of $\repart{B(\BPh)}$, but in general such an offset depends on the array.)

\begin{figure}[ht]
	\begin{minipage}{1.00\columnwidth} 
		
		%	\centering
		
		%		\includegraphics[width = 1.00\columnwidth]{figures/wwb/optimalBeampatternSNRhigh.png}
		\includegraphics[width = 1.00\columnwidth]{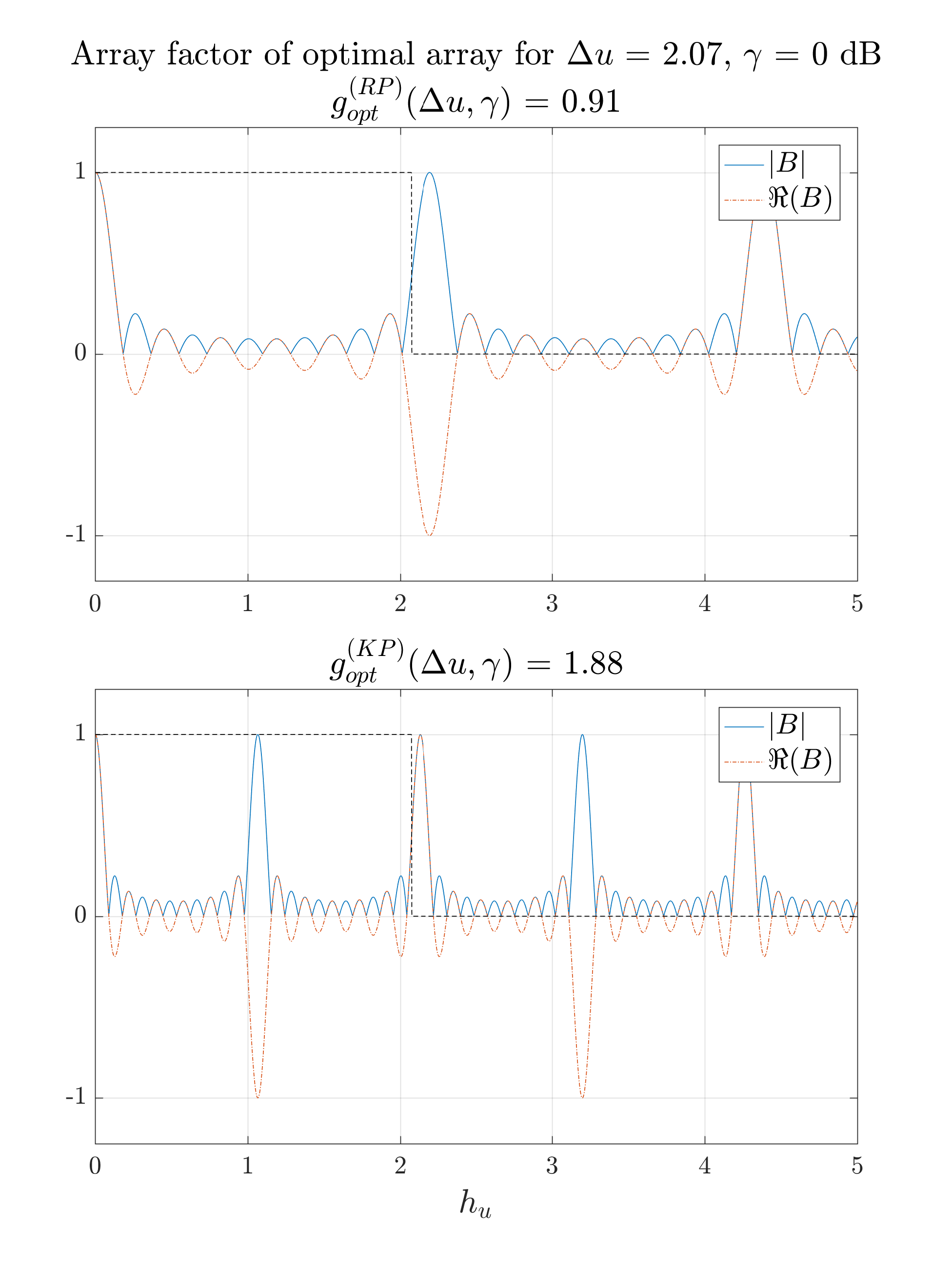}
		\caption{Array factor \eqref{def:Beampattern} of optimal array according to the \gls{WWB} for random-phase and known-phase models for \textit{high} \gls{SNR}. The \textit{alternating symmetry} between the absolute value and the real part only appears when the coordinate origin is the center of mass. The optimal scaling according to \gls{RP} \gls{WWB} seems to depend on $B(\BPh)$ through the absolute value (which is coordinate origin invariant), in contrast with the optimal scaling for the \gls{KP} \gls{WWB} that depends on the real part and thus depends on the coordinate origin.}
		\label{figure:optimalBeampatternSNRhigh}
	\end{minipage}
	\hfill

	\vspace*{0.5cm}
	\begin{minipage}{1.00\columnwidth} 
		
		%	\centering
		
		%		\includegraphics[width = 1.00\columnwidth]{figures/wwb/optimalBeampatternSNRlow.png}
		\includegraphics[width = 1.00\columnwidth]{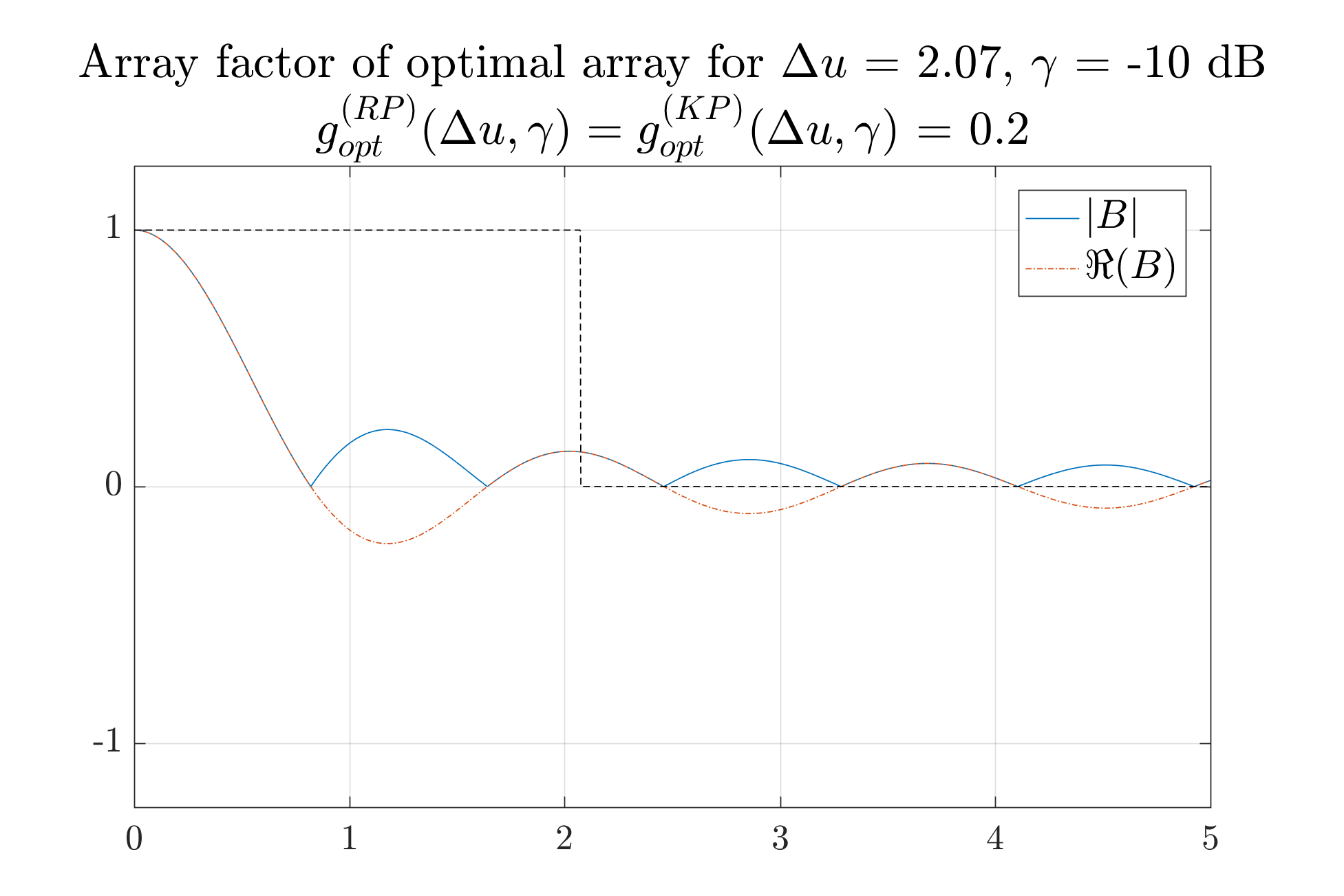}
		\caption{Array factor \eqref{def:Beampattern} of optimal array according to the \gls{WWB} for random-phase and known-phase models. At \textit{low} \gls{SNR}, the optimal scaling is the same for both models.}
		\label{figure:optimalBeampatternSNRlow}
	\end{minipage}
\end{figure}

\section{Adaptive array scaling and channel selection for frequency estimation}\label{sec:adaptive-array-scaling}

Here we apply the adaptive sensing framework of Section~\ref{sec:adaptive-sensing-framework} using the \gls{WWB} metric derived in Section~\ref{sec:wwb-and-models-random-phase-array-processing} to the problem of frequency estimation in two scenarios: (i) adaptation of array \textit{scaling}, and (ii) antenna selection.
In the first case, the parameter optimized is 1-dimensional and we can use the numerical characterization in Section~\ref{sec:controller-wwb}. In the second case, the parameter optimized is discrete, with as many elements as groups of antennas that can be active, and we use a neural network to fit the optimal test point evaluation of the \gls{WWB}.
First we define the Bayesian updates and their particle filter implementation, and then we simulate the closed-loop between the filter and the controller in both scenarios.

\clearpage

\subsection{Bayesian measurement and motion updates}

A characterization of the Bayesian filter requires to define the measurement and the motion updates. 
The likelihood function, $\pd[\obs]{\pargen, \scaling}$ for $\pargen = (\freq, \phase)^T\in\R^{2}$,  required in the measurement update for model~\eqref{eq:ArrayScalingMeasurement}, obeys the general Gaussian model in~\eqref{eq:obs-model-array-processing} 
and is straightforward. 
The transition model between measurement steps is as follows.
The frequency parameter of interest $\freq$ is assumed, in this example, constant ($\freq_k = \freq, \; \forall k$) along the execution of the algorithm $k = 1,2,\dots$, and the initial belief distribution is assumed uniform in the interval $[\freqL, \freqU]$. Naturally, other motion models can be implemented by the particle filter. The initial phase $\phase$ however, undergoes a transition that is crucial for our observation model: after each measurement step it is reinitialized at random, to capture the fact that no information is available due to incoherent measurements. Formally, the initial belief $\pplus_0(\pargen_0)$ is uniform on the Cartesian product $[\freqL, \freqU]\times [-\pi, \pi]$, and 
the state evolution is modeled by 
\begin{align}\label{eq:motion-dynamics-RP}
\pargen_k =
\begin{pmatrix}
\freq_{k-1} \\ 0
\end{pmatrix} + \begin{pmatrix}
0  \\ m_\phase
\end{pmatrix}
\text{ with } m_\phase \sim \Uniform{[-\pi, \pi]} .
\end{align}
This yields transition probabilities independent of the scaling $\scaling_k$, i.e.
\begin{align*}
\transitionden(\pargen_{k} \vert \pargen_{k-1}) = \delta_{\freq_{k-1}}(\freq_{k})\frac{1}{2\pi}\cf{[-\pi, \pi]}(\phase_k) .
\end{align*}

%
%In the next sections we detail the application of the closed-loop Bayesian framework to the above models.
The measurement and motion updates are implemented using a particle filter, described next.

\subsection{Particle filter implementation}\label{sec:characterization-filter}

We employ a particle filter $\{\partPF_{\pargen}, \weightPF\}$ (see e.g. \cite{MSA-SM-NG-TC:02}), 
comprising $\NP$ particles $\partPF_{\pargen} = \{\pargen^i=(\freq^i,\varphi^i)^T \}_{i = 1}^{\NP}\in\R^{\numpar\times \NP}$ and weights $\weightPF\in \R^{\NP\times 1}$ to represent, at each step $k\ge 1$, the belief distribution of $\pargen_k = (\freq_k, \phase_k)^T\in \R^{2}$.
The particle filter is initialized with particles drawn from~$\pplus_0(\pargen_0)\define p_0(\pargen_0)$, i.e. uniformly at random from $[\freqL, \freqU]\times [-\pi, \pi]$, and with equal weights $\weightPF = \frac{1}{\NP}\ones{\NP}$. 
The motion update required to obtain $\pminus_k$ can be realized with the particle filter by applying the target dynamics \eqref{eq:motion-dynamics-RP} independently to each particle. \cg{\margin{, reflecting any particles that cross the boundaries back into the domain $[\freqL, \freqU]$. }} %drawing the phase coordinate of all particles $\{\varphi^i\}_{i = 1}^{\NP}$ anew from a uniform distribution on $[-\pi, \pi]$.
The measurement update is performed by resampling all particles at each step according to the weights given by the likelihood function $w_i = \pd[\obs]{\pargen^i, \sensingpol}$ with the residual resampling method \cite{RD-OC:05}; after the resampling the weights are reset to $\weightPF = \frac{1}{\NP}\ones{\NP}$.

\subsection{Simulation of adaptive array scaling}
Here we compare in simulations the performance of the closed-loop between the Bayesian filter 
and the controllers described in Section~\ref{sec:characterization-controller}
in the sequential frequency estimation task~\eqref{eq:ArrayScalingMeasurement}. The resulting
adaptive strategies are compared with a fixed scaling, a linearly increasing scaling, and a random scaling.

The metric used to evaluate the estimation quality of the policies is the average of squared-errors over
$\Ntrial$ independent trials or executions of the algorithm at a given step $k$, where the ground truth $\freq_0$ is drawn randomly according to the initial belief $p_0(\theta_0)$ at the beginning of each trial,
\begin{align}\label{def:MSEg}
\emse(\freqhat_k) \define \frac{1}{\Ntrial} \sum_{n = 1}^{\Ntrial}((\freqhat_k)_n-(\freq_k)_n)^2 ,
\end{align}
and $\freqhat_k$ is the conditional mean estimator at the $k$-th step, $\freqhat_k = \meanPF_{\freq_k} = \partPF_{\freq_k}\weightPF$.
Further simulation specifications are the following. The target \gls{SNR} is fixed to $\SNR = -5$ dB and assumed known by the controller. To relax this condition a further dimension can be added to the particle filter and then the conditional mean estimate, or a conservative guess, can be used to evaluate the \gls{WWB}.
The sampling vector $\arrI\in \R^{\numobs}$ consists of $\numobs = 12$ uniformly spaced elements.
We employ a particle filter with $\NP = 10^4$ particles to represent the joint belief distribution of the frequency parameter $\freq$ and phase $\phase$. % of which a fraction of $ q_{restart}= 0$ is restarted at each step.\\
The functional dependence of the optimal scaling $\scaling = \scaling(\Delta \freq, \SNR)$ for the \gls{WWB} policies has been computed beforehand on a sufficiently fine grid (Fig.~\ref{figure:GsurfaceSNRpriorRPWWB}). The decision time of the controller is thus made negligible and it is suited for real-time applications. This computation speed is particularly beneficial to analyze the performance in simulations because we find that on the order of $10^4$ trials are required to obtain reproducible results for the empirical \gls{MSE} defined in~\eqref{def:MSEg}.

For the adaptive array scaling estimation task \eqref{eq:ArrayScalingMeasurement}, using as objective function \eqref{eq:rpwwbopt}, in principle the prior can be approximated by the empirical density of the particles.
However, computing the parameter integral via~\eqref{eq:intpriorsGen} for an arbitrary empirical density is expensive due to the large number of particles required for a good representation. For this reason, we approximate the belief distribution represented by the particles $(\partPF_{\pargen}, \weightPF)$ by a uniform or Gaussian distribution of judiciously chosen variance, e.g., in terms of the  empirical variance
%The frequency coordinate of the particles has the following statistics,
%(c.f. Remark \ref{rem:VarianceMapping}) 
%to the empirical variance 
\begin{align}\label{def:empiricalvariance}
\varPF[\freq] &= (\partPF[\freq] - \meanPF[\freq]\ones{\NP})^T \diag(\weightPF) (\partPF[\freq] - \meanPF[\freq]\ones{\NP}) .
%%
%\\
%\meanPF[\freq] &= \partPF[\freq]\weightPF.
%\notag
\end{align}
%
%

%\begin{remark}[Variance mapping]\label{rem:VarianceMapping}
We have observed that the estimation quality of the adaptive sensing policies based on Gaussian and uniform approximations of the empirical density given by the particles can benefit from choosing a larger (i.e., more conservative) variance for the controller input, $\sigma^2_\freq = \delta\cdot  \varPF[\freq]$, i.e., multiplying by a factor the variance of the particles in \eqref{def:empiricalvariance}. 
The choice of $\delta \ge 1$ that works well seems to depend on the \gls{SNR}: For \textit{high} \gls{SNR}, the resulting policies benefit from bigger (more conservative) values. This can be explained based on the abrupt increase of cost reported in  Fig.~\ref{figure:costsurfaceMerge} (left), which reflects the fact that at high SNR there is a possibility of abruptly introducing aliasing in the field of view.
%
%
%which makes the mapping not robust to small errors in the approximation of the belief distribution.
For \textit{low} \gls{SNR}, as in the simulation with $\SNR = -5$ dB, choosing equal variances for the empirical and approximate distribution (i.e. $\delta = 1$) worked fine. This again can be due to the smaller sensitivity of the scaling with respect to the variance at low SNR as shown in Fig.~\ref{figure:costsurfaceMerge} (right) wherein the cost is dominated by sidelobes and not by grating lobes.

Fig.~\ref{figure:StrategiesScaling} shows one realization of scaling choices for each of the strategies, while
Fig.~\ref{figure:MSEcomparison} shows the empirical \gls{MSE} for $10^5$ trials of each of the strategies, confirming the benefit of adaptation strategies over \textit{ad hoc} policies without feedback. For the adaptive policies, we observe the influence of approximating the empirical density of the particles by a Gaussian or uniform prior, which can have a bigger impact than the \gls{SNR} modeling choice that distinguishes the \gls{RP} and \gls{UC} \gls{WWB}.

\begin{figure}[ht]
%\begin{minipage}{0.95\textwidth} 

    \centering

    \includegraphics[width = 0.95\columnwidth]{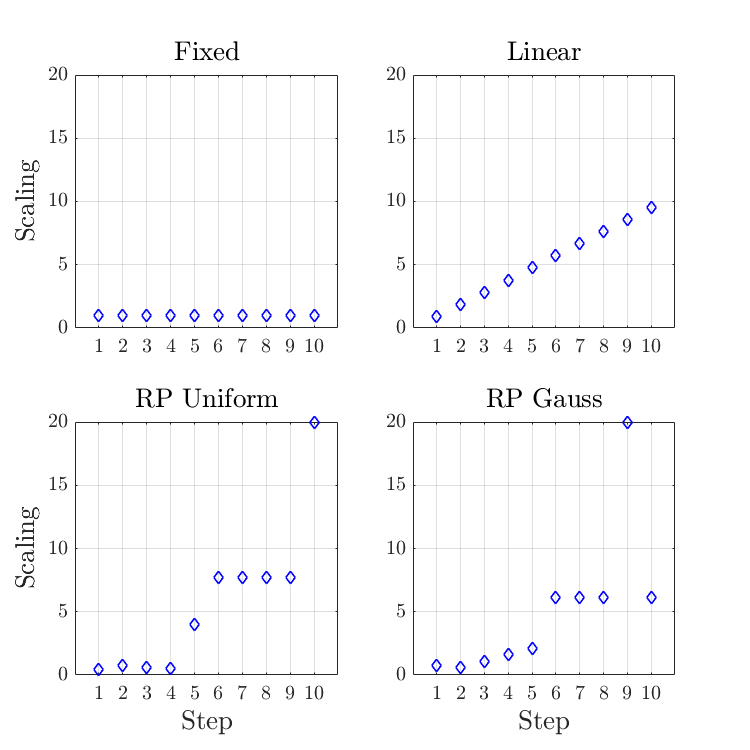}
    \caption{Array scaling over measurement steps for one trial. Top: fixed choice and linearly increasing scaling. Bottom: our adaptive algorithms based on the \gls{RP} \gls{WWB} where the \textit{priors} are given by the output of the particle filter approximated using uniform or Gaussian densities.}
    \label{figure:StrategiesScaling}

\end{figure}

\begin{figure}[ht]
%\begin{minipage}{0.45\textwidth} 

    \centering

    \includegraphics[width = 1.0\columnwidth ]{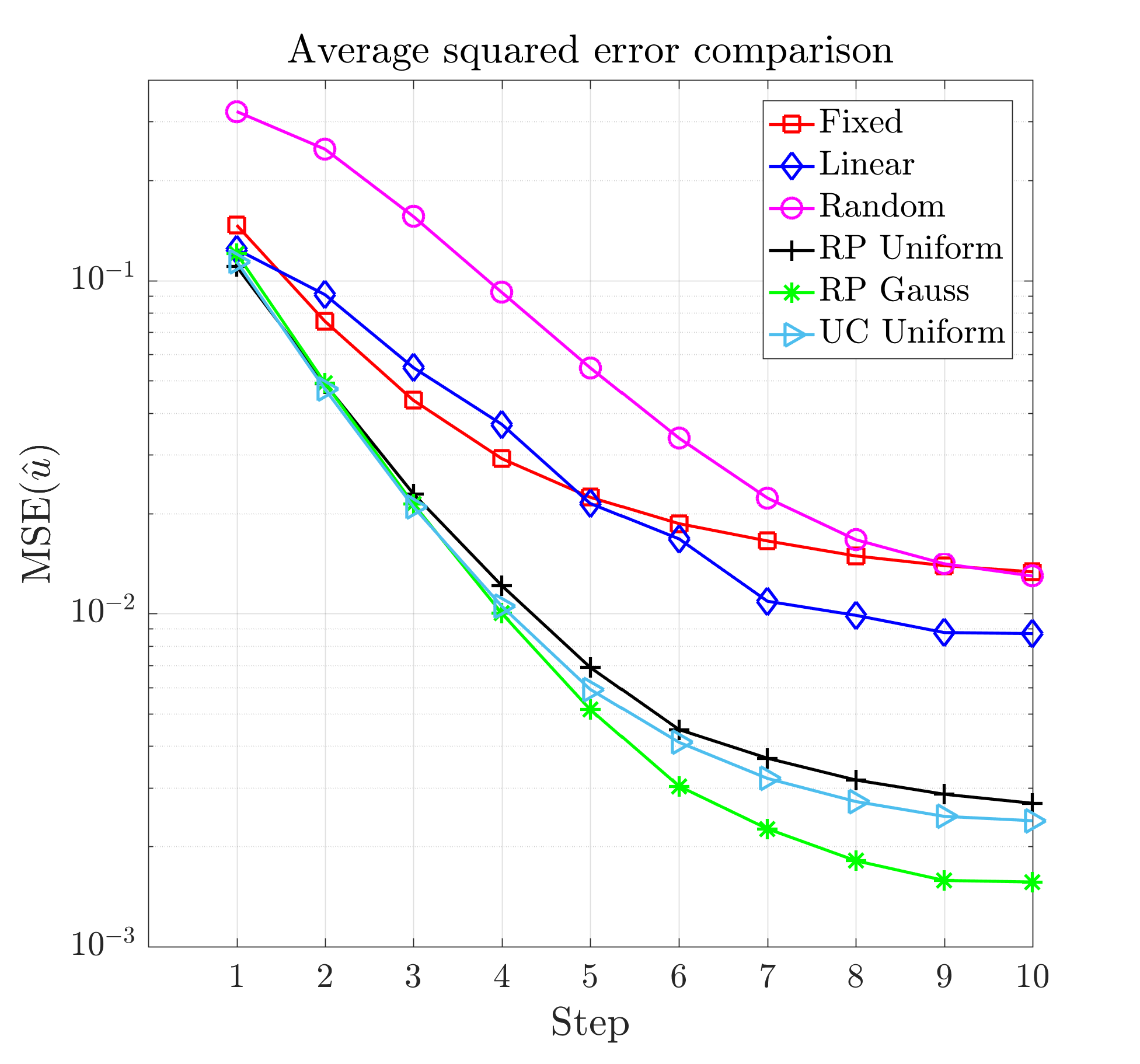}
    \caption{Comparison of \gls{MSE} at each step over $10^5$ independent trials of each policy for SNR $\SNR=-5$ dB. Note that a high number of trials is necessary for this metric to converge because each choice of sensing parameters depends on the filtered belief distribution from the previous step and thus on the unique history of previous choices.}
    \label{figure:MSEcomparison}
%
%\end{minipage}
\end{figure}
We find it interesting to compare, in addition to the average $\emse(\freqhat_k)$, also the %empirical distribution of 
histogram of
the squared errors at each step, cf. Fig.~\ref{figure:MSE_histogram}.
% Fig.~\ref{figure:MSE_histogram} shows the histogram of these squared errors. 
%
It can be seen that the linear scaling strategy often produces estimates equally exact as the adaptive strategy, but is more prone to outliers. 
Conversely, the fixed scaling, which is more conservative, is equally well suited to avoid outliers as the adaptive strategy, but in the prevailing part of trials its estimates are less accurate. 
%This balancing capability highlights the advantage of the adaptive strategy.

 \begin{figure}[ht]
%\begin{minipage}{0.95\columnwidth} 
   % \centering
   %\hspace*{-0.7cm}
 %   \includegraphics[width = 1.00\columnwidth]{figures/arrayscaling/MSE_histogram.png}
        \includegraphics[width = 1.00\columnwidth]{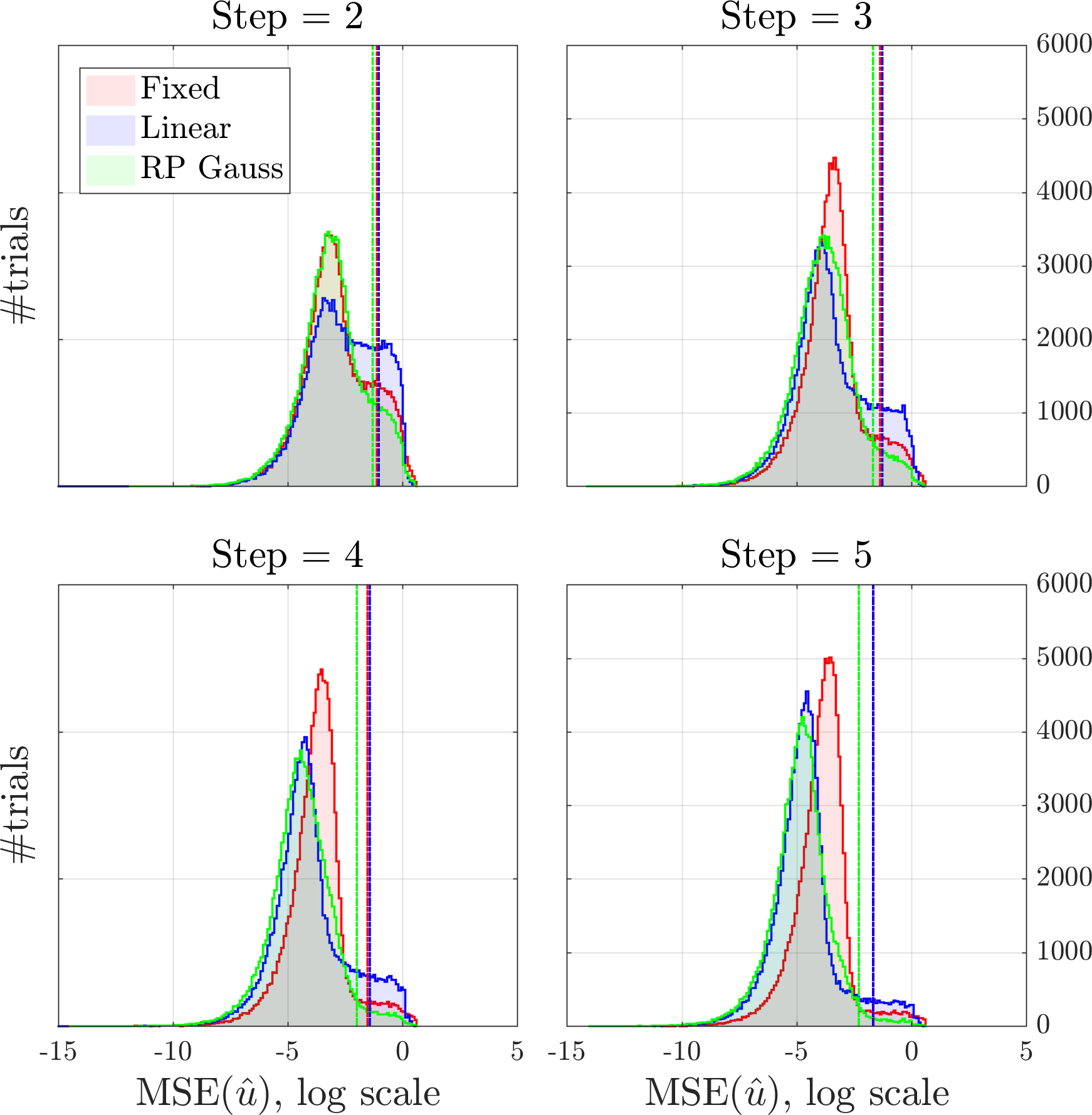}
    \caption{Histogram of errors at given steps of the trial of the algorithm for each policy. The vertical lines represent the \gls{MSE} over all trials at the given step. Note that at the beginning of each trial the ground truth is sampled randomly and therefore this metric resembles the empirical \gls{BMSE}.}
    \label{figure:MSE_histogram}

%\end{minipage}
\end{figure}

\subsection{Simulation of adaptive channel selection}\label{sec:simulation-adaptive-channel-selection}
Here we simulate the controller performance for \gls{DoA} estimation in \gls{TDM} \gls{MIMO}~\eqref{def:TDMmodel} assuming that Doppler is known and equal to $0$. The controller needs to determine at each step the subset of transmitters and receivers that are active~\cite{JT-OI-IB:16}~\cite{DMN-MGH-RS-SB:17}. 

In contrast with the case of array or sampling scaling, where the sensing parameter is one-dimensional and can be computed off-line, stored in a look-up table and interpreted visually, the adaptation of antenna selection presents a number of discrete choices that grows exponentially with the number of available antennas. This has motivated us to train a neural network to predict the values of the evaluation of the tightest 
WWB
%\glslink{WWB}{WW} bound 
over test points in~\eqref{def:costwwb},~cf. Fig.~\ref{fig:adaptive-policy-diagram}. 
% eq:rpwwbopt

We have trained a fully connected neural network to approximate the \gls{KP} \gls{WWB} used in our previous work~\cite{DMN-MGH-RS-SB:17}. 
%under the assumption that the initial phase and the \gls{SNR} are already estimated.
The concept is similar for the newly presented \gls{RP} \gls{WWB}.
 The input data comprises the antenna choices and the variance of the prior distribution, and the output is the optimal \gls{KP} \gls{WWB}~\eqref{def:costwwb}.
%~\eqref{eq:rpwwbopt}. 
The choice of antennas is formatted using $1$-hot encoding of the virtual array elements that are active for a scenario where the available Tx and Rx elements are placed in a uniform grid $0.9 \until{8}$ in units of half-wavelength, and Tx $1$ and Rx $1$ are fixed. That is, at each step the controller chooses one transmitter and one receiver out of $7$ available.
%The other input to the neural net is the variance of the particles that approximate the distribution of the DoA at the output of the particle filter.
% cf. Fig.~\ref{figure:neural-net-wwb}.

For training, we have used the \textit{Tensorflow} library for \textit{Python}. For this small problem, the neural network is allowed to over-fit the training data because we have computed the \gls{WWB} in a sufficiently fine grid of variance values. 
The problem remains for the future to show the application of the closed-loop Bayesian adaptive framework in~Fig.~\ref{fig:adaptive-policy-diagram} to scenarios where the neural network learns to abstract relevant array properties based on limited training data.
%and we have used the tool \textit{Tensorboard} for optimizing hyper-parameters. In this small case we have used

\begin{figure}[ht]
	%\begin{minipage}{0.95\textwidth} 
	
	\centering

		\includegraphics[width = 0.95\columnwidth]{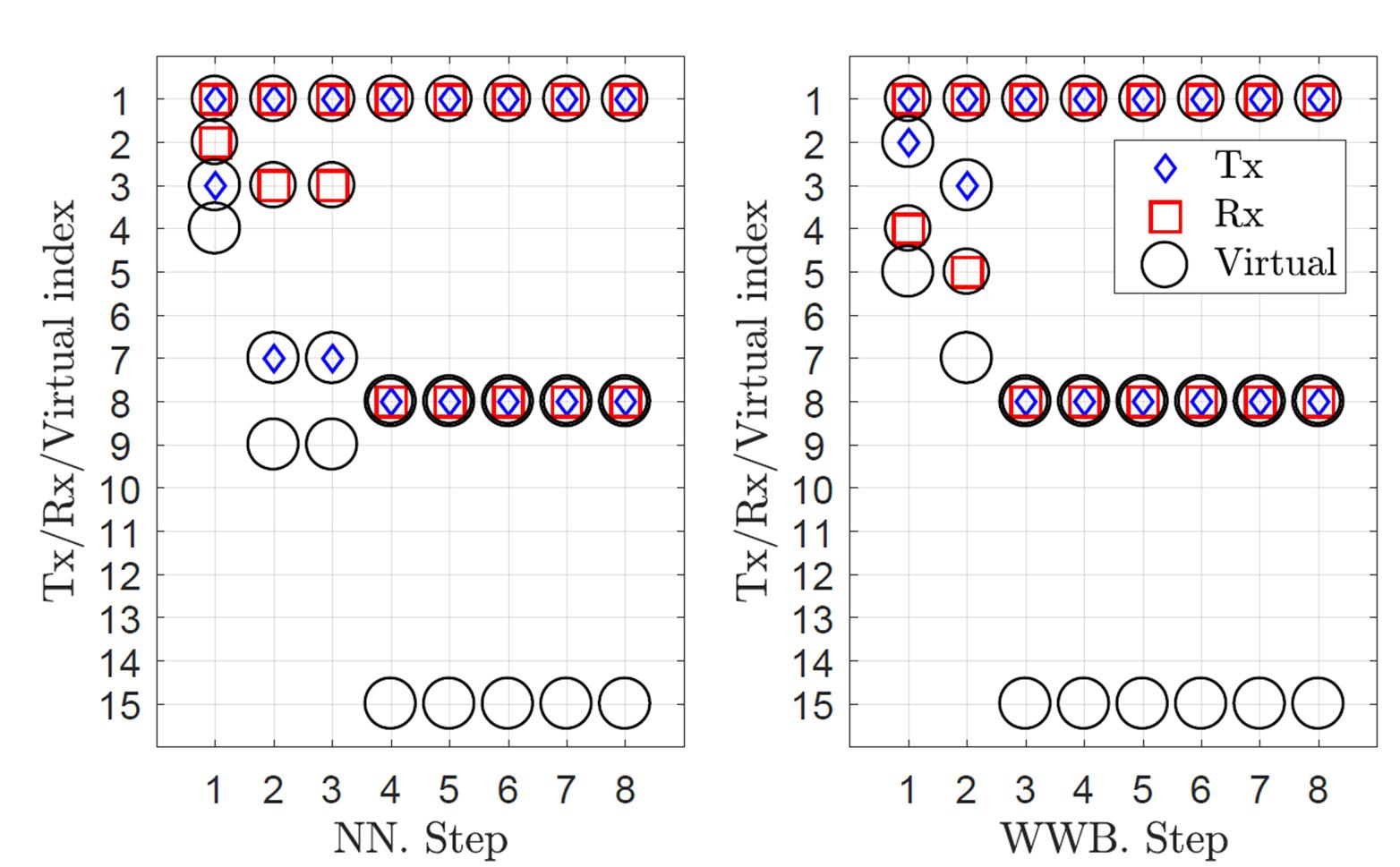}
	\caption{Channel selection at each measurement step for one trial of the policies defined by a fixed choice, the \gls{KP} \gls{WWB}, and a neural network that approximates the \gls{KP} \gls{WWB}. 
		The prior distribution is assumed Gaussian with variance equal to the variance of the distribution given by the particle filter. (Overlapping virtual elements are represented with concentric circles.)}
	\label{figure:choices-nn-wwb}

	%\end{minipage}
	
\end{figure}

\begin{figure}[ht]
	%\begin{minipage}{0.95\textwidth} 
	
	\centering

		\includegraphics[width = 0.95\columnwidth]{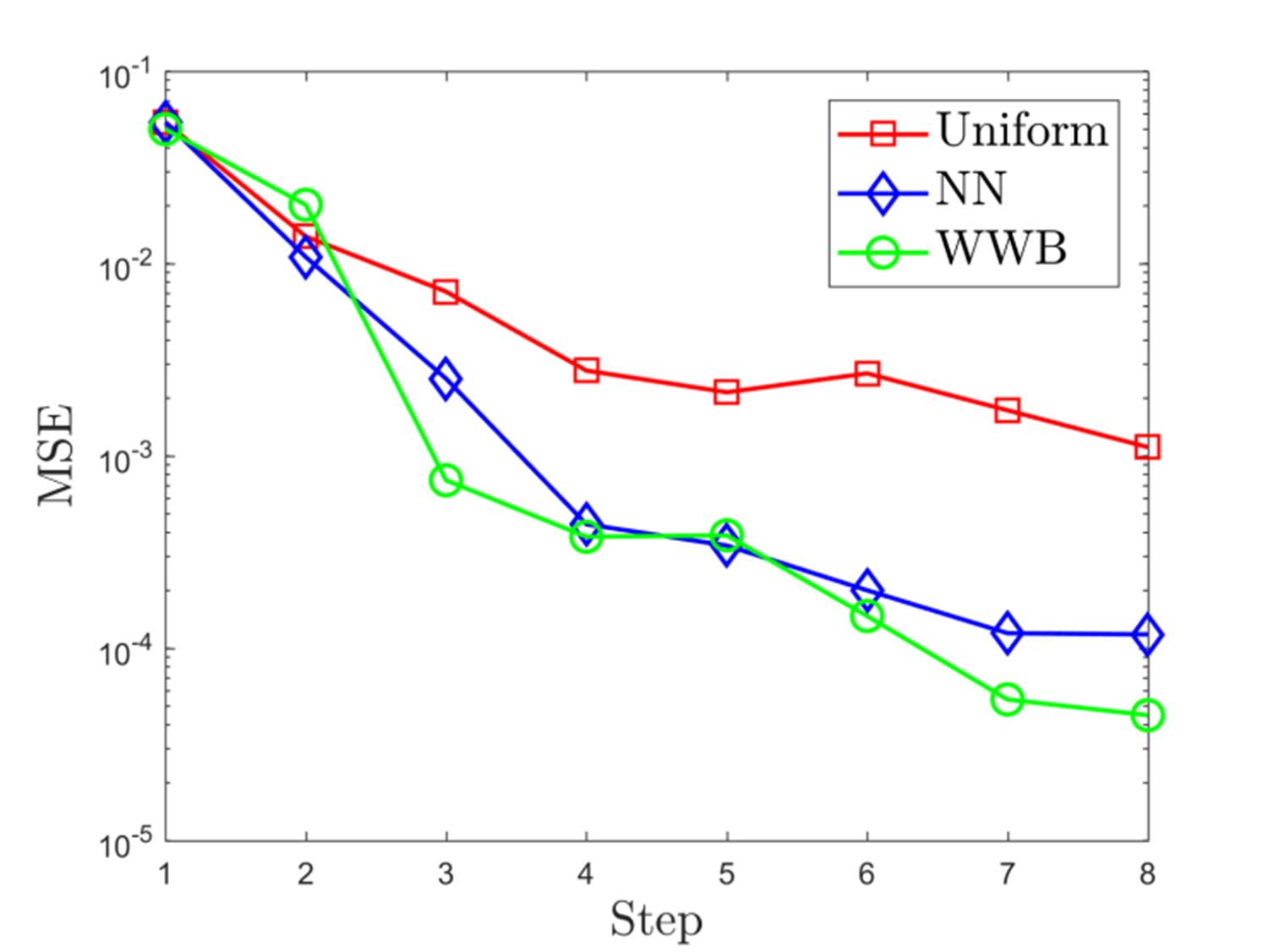}
	\caption{Comparison of \gls{MSE} for policies using the KP WWB~\cite{DTV-AR-RB-SM:14,DMN-MGH-RS-SB:17}, the associated neural network approximation, and the uniform MIMO array with Tx $\{1, 3\}$ and Rx $\{1,2\}$. The MSE is obtained at each step averaging over 300 realizations of the measurement.
%		of the policies in~Fig.~\ref{figure:choices-nn-wwb}.
	}
	\label{figure:mse-channel-choices}

\end{figure}

In Fig.~\ref{figure:choices-nn-wwb} we show that the antenna choices in a typical execution of the neural network resemble the ones of the exact \gls{WWB}, and Fig.~\ref{figure:mse-channel-choices} shows that the performance is similar.
%With this we show that the neural network is able to approximate the channel selection policy based on the WWB, 
From a computation standpoint, this example shows the practical side of the adaptive framework in Fig.~\ref{fig:adaptive-policy-diagram} based on previous work of the authors~\cite{DMN-MGH-RS-SB:17}.

% \begin{figure}[ht]
% 	%\begin{minipage}{0.95\textwidth} 
% 	
% 	\centering
% 	
% 	\includegraphics[width = 0.70\columnwidth]{figures/nn/neuralnet_wwb.jpg}
% 	\caption{}
% 	\label{figure:neural-net-wwb}
% 	%\end{minipage}
% \end{figure}

%\clearpage
%\newpage
%

\section{Conclusions}\label{sec:Conclusion}

We have studied frequency estimation tasks for radar arrays in the context of a Bayesian setting where adaptation of sampling vectors based on the \gls{WWB} prediction of estimation error is shown to be feasible for real-time implementations, at least at the software level, and provides a significant improvement of accuracy. From the ambiguity function standpoint, we have discussed the impact of incorporating knowledge of the phase or lack thereof in the model of the \gls{WWB}, as this modeling choice affects the characterization of aliasing.

We have derived the \glslink{WWB}{Weiss-Weinstein bound} for a generic multi-dimensional frequency estimation model for a single source with random initial phase, which can be efficiently implemented for uniform and Gaussian priors, and stated the optimization problem rigorously to obtain optimal sensing parameters. 
We have shown the applicability in two scenarios of 1D sequential frequency estimation, adapting, respectively, the scaling of sampling vector or \gls{PRF}, e.g., for Doppler estimation, and antenna selection for \gls{DoA} estimation. In the first case we have characterized the optimal controller choices of scaling parameter in terms of prior variance and \gls{SNR}. By storing the values in a look-up table, we can achieve real-time computation. Analogously, in the case of antenna selection for \gls{DoA} estimation, we have shown that a neural net trained off-line can over-fit the predictions of the \gls{WWB} for a given \gls{SNR}, suggesting that the evaluation of the optimal \gls{WWB} is feasible for real-time implementations.

Future work needs to address the bottleneck of the computational cost of the \gls{WWB} for empirical densities, e.g., given by a particle filter. This may be overcome with neural networks trained off-line using as input not the variance but a higher-detail representation of the densities. Through this means, one might obtain well adjusted sensing choices for a larger class of belief distributions than currently possible with Gaussian or uniform approximations of the empirical densities.
 We also envision applications of this framework to scenarios like channel selection in \gls{TDM} \gls{MIMO} for joint \gls{DoA} and Doppler estimation, and \gls{PRF} adaptation for \gls{GMTI} with colored noise. 
Quantitative guarantees on the benefits of adaptation are also an open problem, particularly to complement the need of numerical analysis that requires a large number of Monte Carlo realizations in multi-dimensional problems to extract conclusions about the average behavior of the closed-loop.

\section{Acknowledgments}
The authors would like to thank the German Ministry of Defense, particularly the WTD 81, for supporting this work. We also thank Prof. Joseph Tabrikian for helpful discussions on the topic during the German-Israeli exchange program TA44 and the anonymous reviewers that helped us improve the quality of the paper.

%\bibliographystyle{unsrt} %
%\bibliography{alias,Bib_KR} %,bibliography/Bib_ASP} 

%\vfill
%\pagebreak

\appendix

\section{Appendices}\label{sec:Appendices}
%\section*{Appendices}\label{sec:Appendices}
%\addcontentsline{toc}{section}{Appendices}
%\renewcommand{\thesubsection}{ \Alph{subsection}}

\subsection{Background on the \glslink{WWB}{Weiss-Weinstein bound}}\label{sec:WWB-theory}

For convenience of the reader, we include here the general expression of the \gls{WWB} for Gaussian observations following~\cite{HLVT-KLB:13} and~\cite{DTV-AR-RB-SM:14}. These are the expressions that we explicitly evaluate in Section~\ref{sec:wwb-random-phase} for our array processing models with random initial phase in the case of uniform and Gaussian priors.
The parametric family of \glslink{WWB}{Weiss-Weinstein bounds}  $\wwb(\tpH)\in \R^{\numpar\times \numpar}$ for a data model comprising observations $\obs\in \Obsdom\subseteq\C^\numobs$ and random parameter vector $\pargen\in\R^\numpar$ depends on their joint probability distribution $p(\obs,\pargen)$ and is defined by
	\begin{align}\label{eq:def-wwb-parametric}
\wwb(\tpH)\define \tpH \Q^{-1} \tpH^T ,
	\end{align}
where the elements of the matrix $\Q\in\R^{M\times M}$ are given by \cite{DTV-AR-RB-SM:14}
	\begin{align}\label{eq:def-wwb-Q}
\Q_{k,l} \define \frac{\eta(\tph{k}, \tph{l})+\eta(-\tph{k}, -\tph{l})-\eta(\tph{k}, -\tph{l})-\eta(-\tph{k}, \tph{l})}  {\eta(\tph{k}, \bm{0})\eta(\bm{0}, \tph{l})} .
	\end{align}
The real-valued function $\eta$ is the expectation of scaled "likelihood ratios"  $l(\obs; \tilde{\pargen}, \pargen) \define \frac{\pd{\obs, \tilde{\pargen}}}{\pd{\obs, \pargen }}$ given by
	\begin{align}\label{eq:def-eta-general}
\eta(\tpgenI, \tpgenII) &= \expec{\pd{\obs,\pargen}} [l^\frac{1}{2}(\obs; \pargen+\tpgenI, \pargen)l^\frac{1}{2}(\obs; \pargen+\tpgenII, \pargen) ].
	\end{align}
It is related to the \textit{Bayesian Bhattacharyya coefficient}~\cite{FX-PG-GM-CFM:13} % and the $\alpha$-Chernoff divergence~\cite{FX-PG-GM-CFM:13})
and quantifies the overlap between the shifted densities on the support of the unshifted density.
The matrix of test points has the shape $\tpH = \begin{pmatrix} \tph{1}, ..., \tph{M} \end{pmatrix}\in \R^{\numpar\times M}$
%\begin{align}\label{eq:def-tpH}
%\tpH = \begin{pmatrix} \tph{1}, ..., \tph{M} \end{pmatrix}\in \R^{\numpar\times M}
%\end{align}
for some $M\ge 1$, although $M \geq q$ is recommended in \cite[seq. 4.4.1.4]{HLVT-KLB:13}. 
The domain of valid test points $\tpDom$ is restricted for practical purposes at least to matrices $\tpH$ satisfying 
	\begin{align}\label{eq:domain-condition}
	\tpDom\define \setdef{ \tpH\in\R^{\numpar\times M} }
	 {\Pargen\cap( \Pargen + \tph{m}) \neq \emptyset, \forall m }
	\end{align}
where $	\Pargen \define \operatorname{supp}(\pd{\pargen})=\setdef{\pargen\in\R}{\pd{\pargen}>0}$
%	\begin{align*}
%	\Pargen \define \operatorname{supp}(\pd{\pargen})=\setdef{\pargen\in\R}{\pd{\pargen}>0}
%	\end{align*}
denotes the support of the prior. 
Note that if the intersection of supports in~\eqref{eq:domain-condition} was empty for one $i$, then $\eta(\tph{i}, \bm{0}) = 0$ and therefore $\Q^{-1}$ cannot be computed since the $i$-th row/column of $\Q$ is not defined.

In practice, the joint probability distribution is decomposed as $p(\obs, \pargen) = p(\obs\vert \pargen) p(\pargen)$, because the likelihood function $p(\obs\vert\pargen)$, denoting the probability of the observation $\obs$ given the parameter vector $\pargen$, and the prior probability distribution $p(\pargen)$, can be modeled more naturally. 
With regard to equation~\eqref{eq:def-eta-general}, we find
\begin{align}\label{eq:eta-general}
\eta(\tpgenI, \tpgenII) 
%&= \int_{\Pargen}\int_{\Obsdom} p^\frac{1}{2}(\obs, \pargen+\tpgenI) p^\frac{1}{2}(\obs, \pargen+\tpgenII) d\obs d\pargen \\ \notag
&= \int_{\Pargen} \acute{\eta}_{\pargen}(\tpgenI, \tpgenII) p^\frac{1}{2}(\pargen+\tpgenI)p^\frac{1}{2}(\pargen+\tpgenII)d\pargen
,
\end{align}
where $\acute{\eta}_{\pargen} (\tpgenI, \tpgenII) = \int_{\Obsdom} p^\frac{1}{2}(\obs\vert \pargen+\tpgenI)p^\frac{1}{2}(\obs\vert \pargen+\tpgenII) d\obs.$\\

Up to this point the formulation of the \gls{WWB} applies to general probability distributions of vector parameters and vector observations. Now we consider likelihood functions corresponding to Gaussian observation models parametrized by the mean, as for the \textit{conditional model} described in~\cite{DTV-AR-RB-SM:14}, where $\obs \sim \Gaussian{\C}( \gmodel(\pargen), \noisecov)$
%
%
%\begin{align*}
%\obs \sim \Gaussian{\C}( \gmodel(\pargen), \noisecov)
%\end{align*}
with a known noise covariance matrix $\noisecov$. The authors of \cite[eq. (15)]{DTV-AR-RB-SM:14} offer the following analytic expression for the integration of likelihoods over observation space
\begin{align}\label{eq:def-etaacute-gaussian}
\log \acute{\eta}_{\pargen}(\tpgenI, \tpgenII) = -\frac{1}{4} \Vert \noisecov^{-1/2}(\gmodel(\pargen+\tpgenI) - \gmodel(\pargen+\tpgenII))\Vert^2
\end{align}
which is obtained after using the parallelogram law to the terms that remain after a null addition trick to complete the Gaussian integral.

\subsection{Integral $\intprior$ for  uniform and Gaussian priors }%$\eqref{eq:intpriorsGen}}
\label{sec:IntegralOverPriors} 

Here we give explicit formulas for the \gls{WWB} 
%in~\eqref{eq:rpwwb} 
in~\eqref{eq:randomphase-WWB-general}
for Gaussian and uniform priors providing expressions for~\eqref{eq:intpriorsGen}. These priors can be applied to design problems, such as array design, where the parameter of interest is assumed in a given interval or \glslink{FoV}{Field-of-View} \cite{MGH-DMN-CG-RS:18}. In this work, we use them in our Bayesian adaptive algorithm to approximate the outcome of the particle filter and accelerate the computations of the controller.

\subsubsection{Uniform belief distribution}\label{sec:rpwwbuniform}
Consider a uniform belief distribution with support $\Pargen\subset \R^{\numpar}$ for the parameter of interest, $\pargen \sim \Uniform{\Pargen}$, i.e., $\pd{\pargen} = \frac{1}{\vol{\Pargen}} \cf{\Pargen}(\pargen).$
%\begin{align*}
%\pd{\pargen} = \frac{1}{\vol{\Pargen}} \cf{\Pargen}(\pargen) .
%\end{align*}
We restrict our analysis to independent parameters. This implies a rectangular support
\begin{align}\label{eq:boxdomain}
\Pargen = (\Freqmean + \bigtimes_j [-\frac{\Delta \freq_j}{2}, \frac{\Delta \freq_j}{2}])\times [-\pi, \pi] 
\enifed [\pargenL, \pargenU]
\end{align}
of volume $\vol{\Pargen} = 2\pi\prod_j \Delta \freq_j$ with edge lengths $\Delta \Freq\in \R^{\numpar-1}$ and covariance $\Freqvar = \diag([\frac{\Delta \Freq^2}{12}])$.\\
Using~\eqref{eq:boxdomain}, the integral over priors $\intprior$ in \eqref{eq:intpriorsGen} takes the form
\begin{align}\label{intpriorsUniform}
\intprior(\tpgenI, \tpgenII) &= \frac{ \vol{\Pargenis(\tpgenI, \tpgenII)} }{\vol{\Pargen}} , %\\
%\vol{\Pargenis(\tpgenI, \tpgenII)} & = \vert \Pargen\cap(\Pargen-\tpgenI)\cap(\Pargen-\tpgenII) \vert
\end{align}
where the volume $\vol{\Pargenis(\tpgenI, \tpgenII)}$  of the shifted-support intersection 
\begin{align*}
\Pargenis(\tpgenI, \tpgenII) &\define \Pargen\cap (\Pargen-\tpgenI)\cap (\Pargen-\tpgenII)\\
&= (\max(\pargenL, \pargenL-\tpgenI, \pargenL-\tpgenII), \min(\pargenU, \pargenU-\tpgenI, \pargenU-\tpgenII)) 
\end{align*}
can be readily seen to be
\begin{align}\label{eq:volumeThetatilde}
\!\!\vol{\Pargenis(\tpgenI,\tpgenII)} &\!=\! \prod_j \ramp([\pargenU-\pargenL + \min(\bm{0}, -\tpgenI, -\tpgenII) - \max(\bm{0}, -\tpgenI, -\tpgenII)]_j)
\notag
\\
& \!=\! \prod_j \ramp([\pargenU-\pargenL -\frac{1}{2}(\vert \tpgenI-\tpgenII \vert+\vert \tpgenI \vert+\vert \tpgenII \vert)]_j)
%\prod [\min(\pargenU, \pargenU-\tpgenI, \pargenU-\tpgenII)-\max(\pargenL, \pargenL-\tpgenI, \pargenL-\tpgenII)]_j 
\end{align}
where $\ramp(x)\define\max(0,x)$ is the ramp function, i.e. the product $\vol{\Pargenis(\tpgenI,\tpgenII)}$ must be set to zero if one of the factors is negative. 
\cg{Conditions for this on $(\tpgenI,\tpgenII)$ are ... }
We thus find the expression
\begin{align}\label{eq:rpwwbUniform}
\wwb(\tphI) &= \frac{\tphI \tphI^T}{2\vert \Pargen\vert}  \frac{\acute{\eta}(\tphI,\bm{0})^2\vert \tilde{\Pargen}(\tphI)\vert^2}{\vert \tilde{\Pargen}(\tphI)\vert-\acute{\eta}(\tphI, -\tphI)\vert \tilde{\Pargen}(\tphI, -\tphI)\vert}
\end{align}
with 
\begin{subequations}
\begin{align}
\label{eq:eta-cute1-gaussian}
\acute{\eta}(\tphI, \bm{0}) &= \exp(-\frac{\gamma}{2}(\numobs- \real{\ones{\numobs}^Te^{i\arrex \tphI}}))
\\ 
\label{eq:eta-cute2-gaussian}
\acute{\eta}(\tphI, -\tphI) &= \exp(-\frac{\gamma}{2}(\numobs- \real{\ones{\numobs}^Te^{i2\arrex \tphI}}))
\\
\notag
\vol{\Pargenis(\tphI)}  &= (2\pi-\vert \tphI[\phase]\vert)\prod_j (\Delta \freq_j - \vert \tphI[\freq_j] \vert)\\ \notag
\vol{\Pargenis(\tphI, -\tphI)} &= \max(0,2\pi-2\vert \tphI[\phase] \vert)\prod_j \max(0, \Delta \freq_j - 2\vert \tphI[\freq_j]\vert) .
\end{align}
\end{subequations}

Depending on the shape of the support $\Pargen$, the function $\intprior$ can exhibit certain symmetries.
For the case of our rectangular domain \eqref{eq:boxdomain}, we easily observe $\intprior(\tpgenI, \tpgenII) = \intprior(-\tpgenI, -\tpgenII)$ from the representation in~\eqref{eq:volumeThetatilde}.
As noticed in Remark \ref{rem:symmetryWWB}, the optimization in \eqref{eq:rpwwbopt} can be performed for 
$\tphI \in (\bigtimes_j[-\Delta \freq_j, \Delta \freq_j])\times [0, 2\pi]$.
%\begin{align*}
%\tphI \in (\bigtimes_j[-\Delta \freq_j, \Delta \freq_j])\times [0, 2\pi] .
%\end{align*}

\cg{
\begin{remark}
\item Note that  $ \vol{\Pargenis(\tphI, -\tphI)} = \vol{\Pargenis(2\tphI)}$ for the rectangular support $\Pargen$ of~\eqref{eq:boxdomain}. 
While this fails to hold for general domains (consider a set of two points), it is true at least for convex domains, which is evident from the fact that $(\Pargen+\tphI)\cap (\Pargen- \tphI) \subset \Pargen$ and the translation invariance of the Lebesgue volume. \\
\end{remark}
}

\subsubsection{Gaussian belief distribution}\label{sec:rpwwbgaussian}

Consider a Gaussian belief distribution for the frequency parameter $\Freq \sim \Gaussian{\R}(\Freqmean, \Freqvar)$, which is independent of the uniformly distributed phase $\phase$, i.e.
\begin{align*}
\pd{\pargen} &= \pd{\Freq}\pd{\phase}\\
\pd{\Freq} &= \frac{1}{\sqrt{(2\pi)^{\numpar-1}\det(\Freqvar)}} \exp(-\frac{1}{2} \Vert \Freq -\Freqmean \Vert_{\Freqvar^{-1}})\\
\pd{\phase} &= \frac{1}{2\pi} \cf{[-\pi, \pi]}(\phase) .
\end{align*}
Denoting $\Phase \define [-\pi, \pi]$, the integral over priors $\intprior$ in~\eqref{eq:intpriorsGen} becomes
\begin{align}\label{intpriorsGaussian}
\intprior(\tpgenI, \tpgenII) &= \BC( \tpgenI_{\Freq}, \tpgenII_{\Freq} ) \cdot  \frac{1}{2\pi} \vol{\Phaseis(\tpgenI_{\phase}, \tpgenII_{\phase})}\\ \notag
\BC( \tpgenI_{\Freq}, \tpgenII_{\Freq} ) &= \exp(-\frac{1}{8}\Vert \tpgenI_{\Freq} - \tpgenII_{\Freq} \Vert^2_{\Freqvar^{-1}})\\ \notag
\vol{\Phaseis(\tpgenI_{\phase}, \tpgenII_{\phase})} & = \vert \Phase\cap(\Phase-\tpgenI_{\phase})\cap(\Phase-\tpgenII_{\phase}) \vert\\ \notag
%=  \min(& \pi,  \pi +\tpgenI_{\phase},  \pi + \tpgenII_{\phase}) - \max(-\pi, -\pi+\tpgenI_{\phase}, -\pi + \tpgenII_{\phase})\\ \notag
	%
%&=  2\pi+ \min(0,  \tpgenI_{\phase}, \tpgenII_{\phase}) - \max(0, \tpgenI_{\phase}, \tpgenII_{\phase})\\
	%
&=  \max(0, 2\pi -\frac{1}{2}(\vert \tpgenI_{\phase}-\tpgenII_{\phase}\vert + \vert\tpgenII_{\phase} \vert +\vert\tpgenI_{\phase} \vert))
\end{align}
where we used the expression for the Bhattacharyya coefficient for the case of two Gaussians with same variance but different means~\cite[eq. (61)]{TK67}). (The derivation follows along the same lines as the derivation of the expression for the complex Gaussian likelihood integral in~\eqref{eq:def-etaacute-gaussian}.)
We thus find the expression
\begin{align}\label{eq:rpwwbGaussian}
\wwb(\tphI) &= \\   \notag
 \frac{\tphI \tphI^T}{2(2\pi)}  & \frac{\acute{\eta}(\tphI,\bm{0})^2   (\BC(\tphI_{\Freq},\bm{0})\vert \Phaseis(\tphI_{\phase}) \vert)^2}{\vert \Phaseis(\tphI_{\phase}) \vert - \acute{\eta}(\tphI, -\tphI)  \BC(\tphI_{\Freq},-\tphI_{\Freq}) \vert \Phaseis(\tphI_{\phase}, -\tphI_{\phase}) \vert}
\end{align}
%with $\acute{\eta}(\tphI, \bm{0})$ and $\acute{\eta}(\tphI, -\tphI)$ as
with $\acute{\eta}$ as
 in~\eqref{eq:eta-cute1-gaussian},~\eqref{eq:eta-cute2-gaussian}, and
\begin{align}
	%\notag
	%\acute{\eta}(\tphI, \bm{0}) &= \exp(-\frac{c}{2}(\numobs- \real{\ones{\numobs}^Te^{i\arrex \tphI}}))
	%\\
	%\notag
	%\acute{\eta}(\tphI, -\tphI) &= \exp(-\frac{c}{2}(\numobs- \real{\ones{\numobs}^Te^{i2\arrex \tphI}}))
	%\\
\notag
\vert \Phaseis(\tphI_{\phase}) \vert &= 2\pi - \vert \tphI_{\phase}\vert\\ \notag
\vert \Phaseis(\tphI_{\phase}, -\tphI_{\phase}) \vert &= \max(0, 2\pi - 2\vert \tphI_{\phase}\vert)\\\notag
\BC(\tphI_{\Freq},\bm{0}) &= \exp(-\frac{1}{8}\Vert \tphI_{\Freq}\Vert_{\Freqvar^{-1}}^2)\\\notag
\BC(\tphI_{\Freq},-\tphI_{\Freq}) &= \exp(-\frac{1}{2}\Vert \tphI_{\Freq}\Vert_{\Freqvar^{-1}}^2) .
\end{align}
The optimization~\eqref{eq:rpwwbopt} is for $\tphI \in \R^{\numpar-1}\times [0, 2\pi]$ 
%\begin{align*}
%\tphI \in \R^{\numpar-1}\times [0, 2\pi]
%\end{align*}
since the symmetry $\intprior(\tpgenI, \tpgenII) = \intprior(-\tpgenI, -\tpgenII)$ noticed in Remark \ref{rem:symmetryWWB} is evident from~\eqref{intpriorsGaussian}.

\subsection{Background on the unconditional WWB}\label{sec:UnconditionalWWB}
For convenience of the reader we include the unconditional \gls{WWB} \cite[eq. (56)]{DTV-AR-RB-SM:14}. %(for $s=1/2$) 
It is based on the model $\obs = \sv(\freq)\amplitude + \noise \in \C^\numobs$, where $\noise$ is standard complex Gaussian noise, the steering vector is $\sv(\freq) = e^{i\arrI\freq}$ and the complex amplitude $\amplitude\sim \Gaussian{\C}(0, \SNRUC)$ is also a Gaussian random variable, i.e. $\magnitude^2\sim \frac{\SNRUC}{2}\chi^2_2 = \mathrm{Exp}(\SNRUC)$ has an exponential distribution. Note that the notion of SNR according to the \textit{\gls{KP}} and \textit{\gls{RP}} models, denoted as $\gamma$, is related to $\SNRUC$ by $\expec{}[\magnitude^2]=\expec{}[\SNR] = \SNRUC$. 
%analogue to the \textit{\gls{KP}} and \textit{\gls{RP}} models with known SNR.
When the belief distribution on $\freq$ is a uniform prior of length $\Delta \freq$, the corresponding WWB reads
\begin{align}
 &\rm{UWWB} (h) = \frac{h^2}{2\Delta \freq}\times
 \\
  \notag
 &
 \frac{(\Delta \freq - \vert h\vert)^2 (1+\frac{\USNR}{4} (N^2-\vert 1_N^T e^{idh}\vert^2))^{-2} }
 {(\Delta u - \vert h\vert) - \max(0, \Delta u - 2\vert h\vert) (1+\frac{\USNR}{4}(N^2-\vert 1_N^T e^{i2dh} \vert^2))^{-1} }
% \\
%&\text{with } \USNR = \frac{\SNRUC^2}{\numobs\SNRUC+1} \notag
\end{align}
where $\USNR = \frac{\SNRUC^2}{\numobs\SNRUC+1}.$ Inspiration for this formula came from eq. \cite[eq. (56)]{DTV-AR-RB-SM:14}, which is the special case for $\Delta u = 2$.
The ramp function $\ramp(x)=\max(0,x)$ is required when optimization is performed over $\tpDom = [-\Delta u, \Delta u]$. %(not mentioned in reference \cite[eq. (56)]{DTV-AR-RB-SM:14}).
Due to the symmetry $h\to -h$, the optimization can be restricted to $[0, \Delta u]$. % $\tpDomPlus = [0, \Delta u]$.

\subsection{Bound on \gls{BMSE} conditioned to previous history}\label{sec:proof-all-previous-measurements}

Here we prove Proposition~\ref{cor:CBMSE}. 
We restate part i) in the following Lemma where we spell out the assumed probability dependencies that hold for our observation and transition models.

\begin{lemma}[Motion and measurement updates under sequence of sensing parameters]\label{lemma:updates-identities}
	
	Let the following assumptions be satisfied
	\begin{itemize}\setlength\itemsep{0.1cm}
		\item [(i)] State independence of previous measurements, i.e. %Markov property of the state evolution
		\begin{align}\label{def:StateEvolutionMarkovProperty}
		\pd[\pargenk]{\pargenkm, \sensingpolk} =\pd[\pargenk]{\pargenkm,\obskmall, \sensingpolkall}
		\end{align}
		\item [(ii)] Conditional independence to next sensing parameter, i.e.
		\begin{align*}
		\pd[\pargenkm]{\obskmall,\sensingpolkmall}
		=
		\pd[\pargenkm]{\obskmall,\sensingpolkall} .
		\end{align*}
		
		\item [(iii)] $\obsk$ is independent of $\sensingpolkmall$, $\obskmall$ given $\pargenk$ and $\sensingpolk$, i.e., 
		\begin{align}\label{eq:indepentend-measurements}
		\pd[\obsk]{\pargenk,\sensingpolk} =\pd[\obsk]{\pargenk,\obskmall, \sensingpolkmall, \sensingpolk}.
		\end{align}
	\end{itemize}
Then the recurrences for the motion and measurement updates~\eqref{eq:motion-update-general} and~\eqref{eq:measurement-update-general} 
satisfy~\eqref{eq:updates-identities}.
\end{lemma}

\begin{proof}
	We carry out the proof by complete induction. We can see the assertions to hold for $k=0$ by definition of $\pzeroplus$. Now suppose it is true for $k-1$.
	To show~\eqref{eq:updates-identities-1}, we note that
	\begin{align*}
	%&\pminusatk(\pargenk\vert \sensingpolk) =
	%		\\
	& \;\;\;\;\int \pd[\pargenk]{\pargenkm, \sensingpolk}  \;\pplusatkm(\pargenkm) d\pargenkm
	\\
	&= \int \pd[\pargenk]{\pargenkm,\obskmall, \sensingpolkall}   \;\pd[\pargenkm]{\obskmall, \sensingpolkmall} d\pargenkm
	\\
%	&= \int \pd[\pargenk]{\pargenkm,\obskmall, \sensingpolkall} 
%	\;\pd[\pargenkm]{\obskmall, \sensingpolkall} d\pargenkm
%	\\
	&= \int \pd[\pargenk,\pargenkm]{\obskmall, \sensingpolkall }  d\pargenkm 
	= \pd[\pargenk]{ \obskmall,\sensingpolkall} .
%	\\
%	&= \pd[\pargenk]{ \obskmall,\sensingpolkall} .
	\end{align*}
	In the first step we use (i) and the hypotheses of induction for $\pplusatkm(\pargenkm)$; afterwards we use assumption (ii) and standard properties of probabilities. %and the Law of total probability.
	To show~\eqref{eq:updates-identities-2}, we note that
	\begin{align*}
     %\pplusatk(\pargenk\vert \sensingpolk)  
     %& =  
     %c\; \pd[\obsk]{\pargenk, \sensingpolk} \;\pminusatk(\pargenk\vert \sensingpolk)  
	\pplusatk(\pargenk)  
%& =  
%c\; \pd[\obsk]{\pargenk, \sensingpolk} \;\pminusatk(\pargenk)  
%\\
	&= c\; \pd[\obsk]{\pargenk, \sensingpolk} \pd[\pargenk]{ \obskmall,\sensingpolkall} 
	\\
	&= c\; \pd[\obsk]{\pargenk,\obskmall,\sensingpolkall }
	\pd[\pargenk]{ \obskmall,\sensingpolkall} 
	\\
	&= \tilde{c}\; 
	\pd[\pargenk]{\obsk, \obskmall,\sensingpolkall}
	\end{align*}
	where in the 2nd step we have used assumption (iii), and
	in the last step we have used Bayes rule applied to the probability $\tilde{p}(\obsk\vert \pargenk)\define\pd[\obsk]{\pargenk,\obskmall,\sensingpolkall }$, resulting in  $\tilde{c} = c\cdot \pd[\obsk]{\obskmall,\sensingpolkall}.$
Since the normalizing constant $c$ is chosen so that $\pplusatk$ is a probability density, it is $\tilde{c} = 1$, and
we obtain the value for $c$ in~\eqref{eq:updates-identities-3}.
\end{proof}

Next we observe that the sensing parameters optimized according to~\eqref{def:costwwb} satisfy condition (ii) in~Lemma~\ref{lemma:updates-identities}.

\begin{remark}[]\label{remark:policy-independent-theta}
	Consider the selection of sensing parameters according to~\eqref{def:costwwb} and \eqref{eq:opt-sensing}. Then, for $k\ge1$, it holds that
	\begin{align*}
	\pd[\pargenkm]{\obskmall,\sensingpolkmall}
	=
	\pd[\pargenkm]{\obskmall,\sensingpolkall},
	\end{align*}
	where for $k=1$ it us understood that $p(\pargen_0)
	=
	\pd[\pargen_0]{\sensingpol_1}$.\\
	
	This follows from the fact that $\sensingpolk$ is computed in a deterministic manner in~\eqref{eq:opt-sensing} from the previous observations $\obskmall$ and sensing parameters $\sensingpolkmall$ (requiring in addition only the initial belief $p(\pargen_0)=p_0(\pargen_0)$ and the transition and measurement models to carry out the recurrences~\eqref{eq:motion-update-general} and~\eqref{eq:measurement-update-general}).
	%\eqref{def:recurrences}). 
	%
	As such, $\sensingpolk$ is independent of every other random variable conditioned to $\obskmall$ and $\sensingpolkmall$, and is in particular independent of $\pargenkm$.
%	That is, the belief of $\pargenkm$ is not affected by the choice of the next sensing parameter $\sensingpolk$.
\end{remark}

Using the previous results we can provide the proof of Proposition~\ref{cor:CBMSE}. Proof of i) follows from Lemma~\ref{lemma:updates-identities}. 
\begin{proof}[Proof of ii) in Proposition \ref{cor:CBMSE}]
	Follows using the identity in~\eqref{eq:updates-identities-1} for $\pminusatk(\pargenk)$ %$\pminusatk(\pargenk\vert\sensingpol)$
	 (see Lemma~\ref{lemma:updates-identities}) in inequality~\eqref{eq:ineq-bmse-wwb-general}. % and the Markov property of the state evolution given by {def:StateEvolutionMarkovProperty}.
	Note that the assumptions (i) and (iii)
	of Lemma~\ref{lemma:updates-identities}
	are satisfied for the observation and transition models considered, and condition (ii)
	is verified in Remark~\ref{remark:policy-independent-theta}. 
\end{proof}

% or


\end{document}